\documentclass[11pt]{article}
\usepackage[colorlinks,citecolor=blue,bookmarks=true]{hyperref}
\usepackage{amsmath} 
\usepackage{amsthm} 
\usepackage{amssymb}	
\usepackage{graphicx} 
\usepackage{multirow}
\usepackage{color}
\usepackage[dvips,letterpaper,margin=1in]{geometry}
\usepackage[capitalize,noabbrev]{cleveref}
\usepackage{subcaption}

\usepackage[utf8]{inputenc}
\usepackage[english]{babel}

\usepackage{algorithm}
\usepackage{algpseudocode}
\makeatletter
\renewcommand{\ALG@name}{Algorithm}
\makeatother

\usepackage{bm}
\usepackage{bbm}
\usepackage{diagbox}
\usepackage{mathtools}
\usepackage{ytableau}

\newtheorem{theorem}{Theorem}[section]

\newtheorem{lemma}[theorem]{Lemma}
\newtheorem{corollary}[theorem]{Corollary}

\newtheorem{definition}[theorem]{Definition}

\newtheorem{remark}[theorem]{Remark}

\crefname{question}{Question}{Questions}

\newtheorem{notation}[theorem]{Notation}

\newcommand{\braket}[2]{\left< #1 \vphantom{#2} \middle| #2 \vphantom{#1} \right>} 

\DeclarePairedDelimiter\abs{\lvert}{\rvert}

\DeclarePairedDelimiter\ket{\lvert}{\rangle}
\DeclarePairedDelimiter\bra{\langle}{\rvert}

\newcommand{\E}{\mathop{\bf E\/}}

\newcommand{\tr} {\operatorname{tr}}

\newcommand{\Co}{\mathbb C}

\newcommand{\ketbra}[2]{\ensuremath{\ket{#1}\!\bra{#2}}}
\renewcommand{\braket}[2]{\ensuremath{\langle {#1} \vert {#2} \rangle}}
\newcommand{\kett}[1]{|#1\rangle\!\rangle}
\newcommand{\bbra}[1]{\langle\!\langle#1|}
\newcommand{\kettbbra}[2]{\ensuremath{\kett{#1}\!\bbra{#2}}}
\newcommand{\bbrakett}[2]{\ensuremath{\langle\!\langle{#1}\vert{#2}\rangle\!\rangle}}

\allowdisplaybreaks

\DeclarePairedDelimiter\parens{\lparen}{\rparen}
\DeclarePairedDelimiter\norm{\lVert}{\rVert}
\DeclarePairedDelimiter\braces{\lbrace}{\rbrace}
\DeclarePairedDelimiter\bracks{\lbrack}{\rbrack}

\newcommand{\calA}{\mathcal{A}}

\newcommand{\calE}{\mathcal{E}}
\newcommand{\calF}{\mathcal{F}}

\newcommand{\calV}{\mathcal{V}}

\newcommand{\qchannel}{\textbf{\textup{QChan}}}
\newcommand{\isochannel}{\textbf{\textup{ISO}}}
\newcommand{\dilation}{\textbf{\textup{Dilation}}}
\newcommand{\contract}{\textbf{\textup{Contract}}}

\begin{document}

\title{Quantum channel tomography and estimation by local test}
\author{Kean Chen \thanks{University of Pennsylvania, Philadelphia, USA. Email: \texttt{keanchen.gan@gmail.com}}\and
Nengkun Yu \thanks{Stony Brook University, NY, USA. Email: \texttt{nengkunyu@gmail.com}}\and
Zhicheng Zhang \thanks{University of Technology Sydney, Sydney, Australia. Email: \texttt{iszczhang@gmail.com}}
}
\date{}

\maketitle

\begin{abstract}
We study the estimation of an unknown quantum channel $\mathcal{E}$ with input dimension $d_1$, output dimension $d_2$ and Kraus rank at most $r$. 
We establish a connection between the query complexities in two models: (i) access to $\mathcal{E}$, and (ii) access to a random dilation of $\mathcal{E}$.
Specifically, we show that for parallel (possibly coherent) testers, access to dilations does not help.
This is proved by constructing a local tester that uses $n$ queries to $\mathcal{E}$ yet faithfully simulates the tester with $n$ queries to a random dilation.
As application, we show that:
\begin{itemize}
\item $O(rd_1d_2/\varepsilon^2)$ queries to $\mathcal{E}$ suffice for channel tomography to within diamond norm error $\varepsilon$.
\end{itemize}
Moreover, when $rd_2=d_1$, we show that the Heisenberg scaling $O(1/\varepsilon)$ can be achieved, even if $\mathcal{E}$ is not a unitary channel:
\begin{itemize}
\item 
$O(\min\{d_1^{2.5}/\varepsilon,d_1^2/\varepsilon^2\})$ queries to $\mathcal{E}$ suffice for channel tomography to within diamond norm error $\varepsilon$, and $O(d_1^2/\varepsilon)$ queries suffice for the case of Choi state trace norm error $\varepsilon$. 
\item $O(\min\{d_1^{1.5}/\varepsilon,d_1/\varepsilon^2\})$ queries to $\mathcal{E}$ suffice for tomography of the mixed state $\mathcal{E}(\ketbra{0}{0})$ to within trace norm error $\varepsilon$. 
\end{itemize}
\end{abstract}

\section{Introduction}
Characterizing quantum dynamics is fundamental to quantum computing and quantum information science, playing a central role in the modeling, control, and verification of quantum systems. 
An important question is how to estimate a quantum physical process when it is given as a black box.
\textit{Quantum channel tomography} (also known as \textit{quantum process tomography}) refers to this task: one probes the unknown quantum process, which is mathematically modeled as a quantum channel $\mathcal{E}$, and aims to reconstruct $\mathcal{E}$ from the experimental data.

A notable special case arises when the unknown channel has input dimension $1$, in which case the problem reduces to \textit{quantum state tomography}.
For pure states, the optimal theory has been well understood since the early days of quantum information science~\cite{hayashi1998asymptotic,bruss1999optimal,keyl1999optimal}.
The mixed-state case was settled decades later~\cite{Haah_2017,10.1145/2897518.2897544}, 
using techniques quite different from those developed for pure states (see also \cite{o2017efficient,GKKT20,yuen2023improved,scharnhorst2025optimal,pelecanos2025debiased} for subsequent advances). 
Surprisingly, Pelecanos, Spilecki, Tang, and Wright~\cite{pelecanos2025mixedstatetomographyreduces} recently showed that mixed-state tomography can be reduced to pure-state tomography while achieving the optimal performance. 
This directly inspires us to study efficient methods for quantum channel tomography by leveraging results from isometry channel tomography.

Tomography of a general quantum channel, however, is substantially more challenging, and has been studied for nearly three decades~\cite{chuang1997prescription,poyatos1997complete,leung2000towards,d2001quantum,mohseni2008quantum,kliesch2019guaranteed,bouchard2019quantum,surawy2022projected,Oufkir_2023,oufkir2023adaptivity,huang2023learning,pmlr-v195-fawzi23a,caro2024learning,rosenthal2024quantum,zhao2024learning,zambrano2025fast,yoshida2025quantum}. This added difficulty stems from both the potential power of sequential and adaptive strategies and the complexity of the metrics used to quantify distances between quantum channels.
More specifically, one may prepare arbitrary input states, including states produced by applying the unknown channel in earlier rounds of an experiment. Moreover, standard notions of distance between channels, most notably the diamond norm~\cite{AKN98,watrous2018theory}, are defined via a maximization over all possible input states (including those entangled with an ancilla), which makes both analysis and estimation considerably more demanding.

Haah, Kothari, O'Donnell, and Tang~\cite{haah2023query} resolved the problem for unitary channels by establishing a query complexity of $\Theta(d^2/\varepsilon)$, where $\varepsilon$ denotes the target accuracy in the diamond norm. For quantum channels with input dimension $d_1$ and output dimension $d_2$,
Oufkir addressed the diamond norm tomography problem in the setting of non-adaptive incoherent measurements, showing that the query complexity is $\widetilde{\Theta}(d_1^3 d_2^3/\varepsilon^2)$~\cite{Oufkir_2023,oufkir2023adaptivity}, where the upper bound is by generalizing the process tomography algorithm of~\cite{surawy2022projected}. 
For general measurement schemes, Rosenthal, Aaronson, Subramanian, Datta, and Gur~\cite{rosenthal2024quantum} proved a lower bound of $\Omega(d_1^2 d_2^2/\log(d_1d_2))$ for full Kraus-rank channel tomography; Yoshida, Miyazaki, and Murao~\cite{yoshida2025quantum} proved a lower bound of $\Omega((d_2-d_1)d_1 /(\varepsilon^{2}\log 1/\varepsilon))$\footnote{Since we consider the tomography with success probability at least $2/3$, the lower bound in \cite{yoshida2025quantum} applies to our setting if the success probability is amplified to $1-O(\varepsilon^2)$, which incurs an additional logarithmic factor on $\varepsilon$.} for isometry channel tomography. Notably, these lower bounds hold even if $\varepsilon$ is an average-case distance error.

\subsection{Our results}
In this paper, we study general estimation tasks of quantum channels and establish a connection between the query complexities in different access models.
Our main result is as follows.
\begin{theorem}[Dilations do not help for parallel testers, \cref{coro-12122150} restated]
\label{thm:main-theorem}
If there exists a parallel (possibly coherent) tester that solves a channel estimation task using $n$ queries to an arbitrary dilation of an unknown quantum channel $\mathcal{E}$, then there exists a parallel tester that solves this task using $n$ queries to $\mathcal{E}$ itself.
\end{theorem}

\Cref{thm:main-theorem} provides a clean and systematic approach for designing new quantum algorithms for quantum channel estimation. In general, access to a Stinespring dilation of a quantum channel $\calE$ appears more powerful than access to $\mathcal{E}$ itself.  
However, \Cref{thm:main-theorem} shows that when we restrict to algorithms that make queries in parallel, these two access models are equally powerful in terms of query complexity.
Consequently, one can first design an algorithm assuming queries to the dilation and then translate it into an algorithm that queries the original channel using \cref{thm:main-theorem}. This is often simpler, since the dilation is an isometry and therefore shares many useful properties with unitary operators.

\Cref{thm:main-theorem} can be viewed as an extension of the previous result by \cite{chen2024local}, which studies the power of local test for quantum states and shows that access to purifications does not help for mixed-state testing.
Related results trace back to~\cite[Theorem 35]{soleimanifar2022testing}, and were recently strengthened in an algorithmic sense by \cite{tang2025conjugate}, which explicitly constructs an algorithm for generating random purifications of a mixed state (see also \cite{girardi2025random}). This is further leveraged in~\cite{pelecanos2025mixedstatetomographyreduces} for optimal mixed-state tomography. 
Intuitively, local test and random purification can be viewed as dual concepts in the Heisenberg and Schr{\"o}dinger pictures, respectively.
Finally, we note that \cref{thm:main-theorem} partially answers a conjecture from~\cite{tang2025conjugate} asserting that access to channel dilations does not help.

Now, we introduce the applications of \cref{thm:main-theorem}. Suppose $\mathcal{E}$ is an unknown quantum channel that has input dimension $d_1$, output dimension $d_2$ and Kraus rank at most $r$.

\begin{corollary}[Channel tomography in diamond norm, \Cref{coro-12120134} restated]\label{coro-12131838}
Tomography of $\mathcal{E}$ to within diamond norm error $\varepsilon$ can be done using $O(rd_1d_2/\varepsilon^2)$ queries to $\mathcal{E}$.
\end{corollary}
When $r=d_1d_2$ and $\varepsilon=\Omega(1)$, the result in \cref{coro-12131838} matches the lower bound $\Omega(d_1^2 d_2^2 / \log (d_1 d_2))$ in \cite{rosenthal2024quantum} up to logarithmic factors; and when $r=O(1)$ and $d_2=(1+\Omega(1))d_1$, it matches the lower bound $\Omega((d_2-d_1)d_1/(\varepsilon^2\log 1/\varepsilon))$ in \cite{yoshida2025quantum} up to logarithmic factors. 
\cref{coro-12131838} is obtained by applying \cref{thm:main-theorem} to a diamond-norm isometry channel tomography algorithm (see \cref{thm-1240018}), which is essentially a slight modification of the $O(d^2/\varepsilon^2)$ unitary channel tomography algorithm in~\cite{haah2023query}. 

Notably, when $rd_2=d_1$ (i.e., when the quantum channel $\mathcal{E}$ can be obtained from a unitary channel followed by tracing out a subsystem), we can achieve the Heisenberg scaling $O(1/\varepsilon)$, even if $\mathcal{E}$ is not a unitary channel.
We call the parameter regime $rd_2=d_1$ as the boundary regime, since any quantum channel must satisfy the constraint $rd_2\geq d_1$.

\begin{corollary}[Channel tomography in Heisenberg scaling, \Cref{coro-12120135} restated]\label{coro-12131442}
When $rd_2=d_1$, tomography of $\mathcal{E}$ to within diamond norm error $\varepsilon$ or within Choi state trace norm error $\varepsilon$ can be done using $O(\min\{d_1^{2.5}/\varepsilon,d_1^2/\varepsilon^2\})$ or $O(d_1^2/\varepsilon)$ queries to $\mathcal{E}$, respectively.
\end{corollary}
\cref{coro-12131442} is obtained by applying \cref{thm:main-theorem} to the unitary estimation algorithm due to Yang, Renner, and Chiribella~\cite{yang2020optimal}, which achieves the Heisenberg scaling using parallel queries (see also \cite{kahn2007fast,yoshida2025asymptotically}).

\begin{corollary}[State tomography with state-preparation channels, \Cref{coro-12120136} restated]\label{coro-12131443}
When $rd_2=d_1$, tomography of the mixed state $\mathcal{E}(\ketbra{0}{0})$ to within trace norm error $\varepsilon$ can be done using $O(\min\{d_1^{1.5}/\varepsilon,d_1/\varepsilon^2\})$ queries to $\mathcal{E}$.
\end{corollary}
\cref{coro-12131443} is obtained by applying \cref{thm:main-theorem} to the state estimation algorithm due to Chen~\cite{chen2025inverse}, which achieves the Heisenberg scaling using parallel queries. 

\subsection{Discussion}

While preparing this manuscript, we became aware of an independent and concurrent work by Mele and Bittel~\cite{mele2025optimallearningquantumchannels}, who established the same upper bound $O(rd_1d_2/\varepsilon^2)$ for quantum channel tomography in diamond norm error. 
They also provided an explicit and non-trivial dependence on the failure probability.
We note that their method and ours are based on different techniques.
They obtain the upper bound by analyzing tomography of Choi states, which yields an explicit tomography algorithm.
In contrast, our approach is based on simulating access to dilations of quantum channels using the local test techniques, thereby reducing general channel tomography to a more tractable task---isometry channel tomography. 
In addition, we also provide upper bounds with Heisenberg scaling $O(1/\varepsilon)$ for (possibly non-unitary) channels in the boundary regime $rd_2=d_1$ (i.e., the regime in which channels admit unitary Stinespring dilations).

\subsubsection{Open questions}
\Cref{coro-12131442} shows that the error dependence can be improved from the classical scaling $O(1/\varepsilon^2)$ to the Heisenberg scaling $O(1/\varepsilon)$ in the boundary regime $rd_2= d_1$. 
This raises a new question: how does the true query complexity of quantum channel tomography behave in the near-boundary regime $rd_2\approx d_1$? We conjecture that the transition between classical and Heisenberg scalings is ``smooth'', in the sense that the complexity exhibits a mixture of these two scalings.



Finally, we note that this paper shows dilations do not help for parallel testers, while the conjecture in \cite{tang2025conjugate} remains open in full generality: can one prove that dilations do not help for arbitrary (e.g., sequential) testers?

\section{Preliminaries}

\subsection{Notation}
We use $\mathcal{L}(\mathcal{H})$ to denote the set of linear operators on the Hilbert space $\mathcal{H}$. Similarly, we use $\mathcal{L}(\mathcal{H}_0,\mathcal{H}_1)$ to denote the set of linear operators from $\mathcal{H}_0$ to $\mathcal{H}_1$. Given two orthonormal bases for $\mathcal{H}_0$ and $\mathcal{H}_1$ respectively, we can represent each linear operator in $\mathcal{L}(\mathcal{H}_0,\mathcal{H}_1)$ by a $\dim(\mathcal{H}_1)\times \dim(\mathcal{H}_0)$ matrix and for such a matrix \(X\), we use \(\kett{X}\in \mathcal{H}_1\otimes\mathcal{H}_0\) to denote the vector obtained by flattening the matrix $X$. It is easy to see the following facts:
\[\kett{\ketbra{\psi}{\phi}}=\ket{\psi}\ket{\phi^*}, \quad\quad\quad \kett{XYZ}=X\otimes Z^\textup{T} \kett{Y},\]
where \(\ket{\phi^*}\) is the entry-wise complex conjugate of \(\ket{\phi}\) w.r.t. to a given orthonormal basis, and $Z^\textup{T}$ is the transpose of the matrix $Z$. 
The inner product can be denoted by \(\bbrakett{X}{Y}=\tr(X^\dag Y)\).
For two linear operators $X,Y$, we use $X\sqsubseteq Y$ to denote that $Y-X$ is positive semidefinite.

\subsection{Quantum channels}

A quantum channel with input dimension $d_1$ and output dimension $d_2$ is described by a linear map $\mathcal{E}:\mathcal{L}(\mathbb{C}^{d_1})\rightarrow\mathcal{L}(\mathbb{C}^{d_2})$ such that $\mathcal{E}$ is completely positive and trace-preserving (see, e.g., \cite{NC10,watrous2018theory,hayashi2017quantum}). 

In the Kraus representation~\cite{kraus1983states}, a quantum channel $\mathcal{E}$ is written as
\[\mathcal{E}(\rho)=\sum_{i=1}^r E_i \rho E_i^\dag,\]
where $E_i: \mathbb{C}^{d_1}\rightarrow\mathbb{C}^{d_2}$ are non-zero linear operators that satisfy $\sum_{i=1}^r E_i^\dag E_i=I$, which are called Kraus operators. We can always find a set of $E_i$ such that $\tr(E_i^\dag E_j)=0$ for $i\neq j$, then those $E_i$ are called orthogonal Kraus operators and $r$ is called the Kraus rank.
Note that $r$ must satisfy the constraint $d_1/d_2\leq r\leq d_1d_2$.
In particular, a quantum channel that has Kraus rank $r=1$ is an isometry channels $\mathcal{V}=V(\cdot) V^\dag$, where $V:\mathbb{C}^{d_1}\rightarrow\mathbb{C}^{d_2}$ is an isometry operator, i.e., $V^\dag V=I_{d_1}$, and it must hold that $d_1\leq d_2$.

\begin{notation}
We use $\qchannel_{d_1,d_2}^r$ to denote the set of all quantum channels $\mathcal{E}:\mathcal{L}(\mathbb{C}^{d_1})\rightarrow \mathcal{L}(\mathbb{C}^{d_2})$ that have Kraus rank at most $r$.
In particular, we use $\isochannel_{d_1,d_2}$ to denote the set of isometry channels with input dimension $d_1$ and output dimension $d_2$, which is equivalent to $\qchannel_{d_1,d_2}^1$. 
\end{notation}

In the Choi-Jamio{\l}kowski representation~\cite{choi1975completely,jamiolkowski1972linear}, $\mathcal{E}$ is represented by the Choi-Jamio{\l}kowski operator 
\[C_\mathcal{E}=(\mathcal{E}\otimes \mathcal{I})(\kettbbra{I}{I})\in\mathcal{L}(\mathbb{C}^{d_2}\otimes\mathbb{C}^{d_1}),\]
where $\kett{I}=\sum\limits_{i}\ket{i}\ket{i}\in \mathbb{C}^{d_1}\otimes \mathbb{C}^{d_1}$ is an unnormalized maximally entangled state. We may simply call it the Choi operator. 
Note that we can write $C_\mathcal{E}=\sum_{i=1}^r\kettbbra{E_i}{E_i}$, where $E_i$ are orthogonal Kraus operators and thus $\kett{E_i}$ are pairwise orthogonal vectors. Therefore, the Kraus rank equals the rank of the Choi operator.

\paragraph{Stinespring dilation.}
Using the Stinespring dilation~\cite{stinespring1955positive}, we can also write a quantum channel $\mathcal{E}$ with Kraus operators $\{E_i\}_{i=1}^r$ as
\begin{equation}\label{eq-1230155}
\mathcal{E}(\rho)=\tr_{\mathcal{H}_\mathrm{anc}}(V\rho V^\dag),
\end{equation}
where $\mathcal{H}_\mathrm{anc}\cong \mathbb{C}^{r}$ and $V=\sum_{i=1}^r\ket{i}_\mathrm{anc}\otimes E_i$ is an isometry operator. An isometry channel $\mathcal{V}=V(\cdot)V^\dag$ that satisfies \cref{eq-1230155} is called a dilation of $\mathcal{E}$. 
Suppose $\mathcal{V}_1$ is a dilation of $\mathcal{E}$, then $\mathcal{V}_2$ is a dilation of $\mathcal{E}$ if and only if they differ by a unitary on $\mathcal{H}_\mathrm{anc}$, i.e., $V_2=(U\otimes I_{d_2}) V_1$ for $U:\mathcal{H}_\mathrm{anc}\rightarrow\mathcal{H}_\mathrm{anc}$ a unitary.

\begin{notation}
For a quantum channel $\mathcal{E}$ with Kraus rank at most $r$, we use $\dilation_{r}(\mathcal{E})$ to denote the set of all dilations of $\mathcal{E}$ with an ancilla system of dimension $r$. 
For an isometry channel $\mathcal{V}:\mathcal{L}(\mathcal{H}_1)\rightarrow\mathcal{L}(\mathcal{H}_2\otimes\mathcal{H}_3)$, we use $\contract_{\mathcal{H}_3}(\mathcal{V})$ to denote the quantum channel
\[\rho\mapsto \tr_{\mathcal{H}_3}(V\rho V^\dag).\]
\end{notation}

\paragraph{Haar distribution.}


Given a quantum channel $\mathcal{E}\in\qchannel_{d_1,d_2}^r$, we define the Haar distribution on $\dilation_r(\mathcal{E})$ by the following procedure: pick an arbitrary dilation $\mathcal{V}\in\dilation_r(\mathcal{E})$ and output $(\mathcal{U}\otimes \mathcal{I}_{d_2})\circ\mathcal{V}$ for a Haar random unitary $U \in\mathbb{U}_r$.
This is well defined since the resulting distribution does not depend on the choice of the dilation $\mathcal{V}$. It is easy to see that this distribution is invariant under $\mathbb{U}_{r}$ in the following sense:
\[\Pr[A]=\Pr[\{(\mathcal{U}\otimes \mathcal{I}_{d_2})\circ\mathcal{V}\,|\, \mathcal{V}\in A\}],\]
for any $U\in\mathbb{U}_r$ and measurable set $A\subseteq \dilation_r(\mathcal{E})$.

\begin{notation}
We use $\mathcal{V}\sim\dilation_r(\mathcal{E})$ and $U\sim \mathbb{U}_d$ to denote that $\mathcal{V}$ and $U$ are sampled from Haar distributions on $\dilation_r(\mathcal{E})$ and $\mathbb{U}_d$, respectively.
\end{notation}

\subsection{Formalism of quantum channel testers}
A \textit{quantum channel tester} means a quantum algorithm that can make multiple queries to an unknown quantum channel and then produces a classical output.
We adopt the quantum tester formalism based on Choi-Jamio{\l}kowski representation (see, e.g., \cite{chiribella2009theoretical,bavaresco2021strict,bavaresco2022unitary}), which provides a practical framework for studying various classes of quantum testers, such as parallel and sequential ones.

First, we define the link product ``$\star$'': 
\begin{definition}[Link product ``$\star$''~\cite{chiribella2008quantum,chiribella2009theoretical}]
\label{def-720255}
Suppose $X$ is a linear operator on $\mathcal{H}_{\bm{i}}=\mathcal{H}_{i_1}\otimes\mathcal{H}_{i_2}\otimes\cdots\otimes\mathcal{H}_{i_{n}}$ and $Y$ is a linear operator on $\mathcal{H}_{\bm{j}}=\mathcal{H}_{j_1}\otimes\mathcal{H}_{j_2}\otimes\cdots\otimes\mathcal{H}_{j_{m}}$,
where $\bm{i}=(i_1,\ldots,i_n)$ is a sequence of pairwise distinct indices, and likewise for $\bm{j}=(j_1,\ldots,j_m)$.
Let $\bm{a}=\bm{i}\cap\bm{j}$ be the set of indices in both $\bm{i}$ and $\bm{j}$ and $\bm{b}=\bm{i}\cup\bm{j}$ be the set of indices in either $\bm{i}$ or $\bm{j}$.
Then, the combination of $X$ and $Y$ is defined by
\[X\star Y= \tr_{\mathcal{H}_{\bm{a}}}\!\left(X^{\textup{T}_{\mathcal{H}_{\bm{a}}}} \cdot Y\right)=\tr_{\mathcal{H}_{\bm{a}}}\!\left(X\cdot Y^{\textup{T}_{\mathcal{H}_{{\bm{a}}}}}\right),\]
where $\mathcal{H}_{\bm{a}}$ means the tensor product of subsystems labeled by the indices in $\bm{a}$, $\textup{T}_{\mathcal{H}_{\bm{a}}}$ means the partial transpose on $\mathcal{H}_{\bm{a}}$, both $X$ and $Y$ are treated as linear operators on $\mathcal{H}_{\bm{b}}$, extended by tensoring with the identity operator as needed.
\end{definition}
\begin{remark}
The link product has many good properties: it preserves the L\"owner order: if $X,Y\sqsupseteq 0$ then $X\star Y\sqsupseteq 0$~\cite[Theorem 2]{chiribella2009theoretical}. It is commutative $X\star Y=Y \star X$, and associative $(X\star Y)\star Z=X\star (Y\star Z)$ whenever $X,Y,Z$ do not share a common subsystem (i.e., there is no subsystem that is a subsystem of all three).
Moreover, it characterizes the channel concatenation under the Choi representation: given two quantum channels $\mathcal{E}_1:\mathcal{L}(\mathcal{H}_1)\rightarrow\mathcal{L}(\mathcal{H}_2)$ and $\mathcal{E}_2:\mathcal{L}(\mathcal{H}_2)\rightarrow\mathcal{L}(\mathcal{H}_3)$, we have 
$C_{\mathcal{E}_2\circ \mathcal{E}_1}=C_{\mathcal{E}_2}\star C_{\mathcal{E}_1}$, where $C_\mathcal{E}$ denotes the Choi operator of $\mathcal{E}$.
\end{remark}

\paragraph{Parallel testers.}
Suppose a quantum channel tester uses $n$ queries to an unknown quantum channel $\mathcal{E}$. 
We label the input and output systems of the $i$-th query to $\mathcal{E}$ as $\mathcal{H}_{\mathrm{A},i}$ and $\mathcal{H}_{\mathrm{B},i}$, i.e., the $i$-th copy of the unknown channel is a linear map from $\mathcal{L}(\mathcal{H}_{\mathrm{A},i})$ to $\mathcal{L}(\mathcal{H}_{\mathrm{B},i})$.

In a parallel tester, one prepares a multipartite input state, possibly including ancilla systems, and applies the unknown channel in parallel to its subsystems, ensuring that the output of any use never interacts with the inputs of the others. After all channel uses, a single joint measurement is performed on the combined output state.

\begin{definition}[Parallel tester]
A parallel tester is a set of linear operators $\{T_i\}_i$ for $T_i\in \mathcal{L}(\bigotimes_{j=1}^n \mathcal{H}_{\mathrm{A},j}\otimes \mathcal{H}_{\mathrm{B},j})$, such that $T_i\sqsupseteq 0$ and $\sum_{i} T_i=\rho_{\mathrm{A}}\otimes I_{\mathrm{B}}$, where $I_{\mathrm{B}}$ is the identity operator on $\bigotimes_{j=1}^n\mathcal{H}_{\mathrm{B},j}$, and $ \rho_{\mathrm{A}}$ is a positive semidefinite operator on $\bigotimes_{j=1}^n\mathcal{H}_{\mathrm{A},j}$ and $\tr(\rho_\mathrm{A})=1$.
\end{definition}

When we apply a parallel tester $\{T_i\}_i$ to a quantum channel $\mathcal{E}$, we get the classical outcome $i$ with probability 
\begin{equation}\label{eq-1211852}
p_i=T_i\star C_\mathcal{E}^{\otimes n}= \tr(T_i (C_{\mathcal{E}}^{\otimes n})^\mathrm{T})=\tr(T_i^\mathrm{T} C_{\mathcal{E}}^{\otimes n}),
\end{equation}
where $C_\mathcal{E}^{\otimes n}\in\mathcal{L}(\bigotimes_{j=1}^n \mathcal{H}_{\mathrm{A},j}\otimes\mathcal{H}_{\mathrm{B},j})$
is the Choi operator of all $n$ queries to the channel $\mathcal{E}$ and $(\cdot)^{\mathrm{T}}$ denotes matrix transposition.

To see that the parallel tester $\{T_i\}_i$ can be realized by an algorithm that makes queries in parallel, we consider the following procedure:
\begin{itemize}
    \item Assume $\sum_i T_i=\rho_\mathrm{A}\otimes I_{\mathrm{B}}$. Prepare a quantum state $(\sqrt{\rho_\mathrm{A}}^\textup{T}\otimes I_{\mathrm{A}})\kett{I_{\mathrm{A}}}$ in $\bigotimes_{j=1}^n\mathcal{H}_{\mathrm{A},j}\otimes\bigotimes_{j=1}^n\mathcal{H}_{\mathrm{A},j}$. Indeed, this is a valid quantum state because
    $\bbra{I_{\mathrm{A}}} (\rho_\mathrm{A}^\mathrm{T}\otimes I_{\mathrm{A}})\kett{I_{\mathrm{A}}}=\tr(\rho_\mathrm{A}^{\textup{T}})=1$.
    
    \item Apply the quantum channel $\mathcal{I}_\mathrm{A}\otimes \mathcal{E}^{\otimes n}$ on the prepared state and obtain the mixed state $(\sqrt{\rho_\mathrm{A}}^\mathrm{T}\otimes I_{\mathrm{B}}) C_\mathcal{E}^{\otimes n} (\sqrt{\rho_\mathrm{A}}^\mathrm{T}\otimes I_{\mathrm{B}})$.

    \item Perform the POVM $\left\{(\sqrt{\rho}_\mathrm{A}^\mathrm{T}\otimes I_{\mathrm{B}})^{-1} \, T_i^\mathrm{T}\, (\sqrt{\rho}_\mathrm{A}^\mathrm{T}\otimes I_{\mathrm{B}})^{-1}\right\}_i$ and obtain the result $i$, where $(\cdot)^{-1}$ is the pseudo-inverse. Then, one can easily see that the probability of getting result $i$ is exactly that in \cref{eq-1211852}.
\end{itemize}
Conversely, any algorithm that makes queries in parallel can be described by a parallel tester.
To see this, assume that the algorithm first prepares a state $\rho$ on $(\bigotimes_{j=1}^n \mathcal{H}_{\mathrm{A},j})\otimes \mathcal{H}_\mathrm{anc}$, where $\mathcal{H}_{\mathrm{anc}}$ is an ancilla system, and then apply the channel $\mathcal{E}^{\otimes n}\otimes \mathcal{I}_\mathrm{anc}$ on $\rho$ followed by a POVM $\{E_i\}_i$, where each $E_i \in\mathcal{L}((\bigotimes_{j=1}^n \mathcal{H}_{\mathrm{B},j})\otimes\mathcal{H}_\mathrm{anc})$ is positive semidefinite. 
Then, let $T_i=E_i^\mathrm{T}\star \rho$. We can see that $\{T_i\}_i$ is a parallel tester and the probability of obtain outcome $i$ is
\[\tr(E_i\cdot (\mathcal{E}^{\otimes n}\otimes\mathcal{I}_\mathrm{anc})(\rho))=\tr(E_i \cdot (C_\mathcal{E}^{\otimes n}\star \rho))=E_i^\mathrm{T}\star C_\mathcal{E}^{\otimes n} \star \rho=T_i\star C_\mathcal{E}^{\otimes n},\]
which is exactly the same as that in \cref{eq-1211852}.



\subsection{Schur-Weyl duality on bipartite systems}
Let $\mathcal{H}_1,\mathcal{H}_2,\ldots,\mathcal{H}_n$ be a sequence of Hilbert spaces such that $\mathcal{H}_i\cong\mathbb{C}^d$ for $1\leq i\leq n$.
Consider the Hilbert space \(\bigotimes_{i=1}^n\mathcal{H}_i\). This space admits representations of the symmetric group \(\mathfrak{S}_n\) (i.e., the group of all permutations on the set $\{1,2,\ldots,n\}$) and unitary group \(\mathbb{U}_d\) (i.e., the group of unitaries on $d$-dimensional Hilbert space). The unitary group acts by simultaneous ``rotation'' as \(U^{\otimes n}\) for any \(U\in \mathbb{U}_d\) and the symmetric group acts by permuting tensor factors:
\begin{equation}
\texttt{p}(\pi)\ket{\psi_1}\cdots\ket{\psi_n}=\ket{\psi_{\pi^{-1}(1)}}\cdots\ket{\psi_{\pi^{-1}(n)}},
\end{equation}
where \(\pi\in\mathfrak{S}_n\). 
Two actions \(U^{\otimes n}\) and \(\texttt{p}(\pi)\) commute with each other, and hence $\bigotimes_{i=1}^n\mathcal{H}_i$ admits a representation of group \(\mathbb{U}_d\times \mathfrak{S}_n\). More specifically, the Schur-Weyl duality (see, e.g., \cite{fulton2013representation}) states that
\begin{equation}\label{eq-1111224}
\bigotimes_{i=1}^n\mathcal{H}_i\stackrel{\mathfrak{S}_n\times\mathbb{U}_d}{\cong}\bigoplus_{\lambda \vdash_d\, n}\mathcal{P}_\lambda\otimes \mathcal{Q}^d_\lambda,
\end{equation}
where \(\mathcal{P}_\lambda\) and \(\mathcal{Q}^d_\lambda\) are irreducible representations of \(\mathfrak{S}_n\) and \(\mathbb{U}_d\) labeled by Young diagram $\lambda$, respectively. We use $\texttt{p}_\lambda(\pi)$ and $\texttt{q}_\lambda(U)$ to denote the actions of $\pi\in\mathfrak{S}_n$ and $U\in\mathbb{U}_d$ on $\mathcal{P}_\lambda$ and $\mathcal{Q}_\lambda^d$, respectively.

Now, suppose we have two sequences of Hilbert spaces $(\mathcal{H}_{\textup{A},1},\ldots,\mathcal{H}_{\textup{A},n})$ and $(\mathcal{H}_{\textup{B},1},\ldots,\mathcal{H}_{\textup{B},n})$, where $\mathcal{H}_{\textup{A},i}\cong \mathbb{C}^{d_1}$ and $\mathcal{H}_{\textup{B},j} \cong \mathbb{C}^{d_2}$. We define the action of group $\mathfrak{S}_{n}\times\mathfrak{S}_n$ on $\bigotimes_{i=1}^n \mathcal{H}_{\textup{A},i}\otimes \mathcal{H}_{\textup{B},i}$ as $\texttt{p}_\mathrm{A}(\pi_1)\otimes \texttt{p}_{\mathrm{B}}(\pi_2)$ for $(\pi_1,\pi_2)\in\mathfrak{S}_n\times\mathfrak{S}_n$, where $\texttt{p}_\mathrm{A}(\cdot)$ denotes the permutation action on $\bigotimes_{i=1}^n \mathcal{H}_{\mathrm{A},i}$ (and similarly for $\texttt{p}_\mathrm{B}(\cdot)$).
We define the action of \(\mathbb{U}_{d_1}\times \mathbb{U}_{d_2}\) on \(\bigotimes_{i=1}^n \mathcal{H}_{\textup{A},i}\otimes \mathcal{H}_{\textup{B},i}\) as $(U_\mathrm{A}\otimes U_{\textup{B}})^{\otimes n}$ for $(U_{\textup{A}},U_{\textup{B}})\in \mathbb{U}_{d_1}\times\mathbb{U}_{d_2}$. 
Note that the action of $\mathfrak{S}_n\times\mathfrak{S}_n$ commutes with the action of $\mathbb{U}_{d_1}\times\mathbb{U}_{d_2}$. Therefore, we have a Schur-Weyl duality on this bipartite system as
\[\bigotimes_{i=1}^n \mathcal{H}_{\textup{A},i}\otimes \mathcal{H}_{\textup{B},i}\,\, \stackrel{\mathfrak{S}_{n}\times\mathfrak{S}_n\times\mathbb{U}_{d_1}\times\mathbb{U}_{d_2}}{\cong}\,\,\bigoplus_{\substack{\lambda\vdash_{d_1} \, n\\ \mu\vdash_{d_2} \, n}} \mathcal{P}_{\lambda}\otimes\mathcal{P}_\mu\otimes \mathcal{Q}^{d_1}_\lambda\otimes\mathcal{Q}^{d_2}_\mu,\]
where $\mathcal{P}_\lambda\otimes\mathcal{P}_\mu\otimes\mathcal{Q}^{d_1}_\lambda\otimes\mathcal{Q}^{d_2}_\mu$ is an irreducible representation of $\mathfrak{S}_n\times\mathfrak{S}_n\times \mathbb{U}_{d_1}\times\mathbb{U}_{d_2}$.

\section{Local test of quantum channel}
In this section, we introduce and prove our main results.

\begin{theorem}\label{thm-1221159}
Let $d_1,d_2,r$ be positive integers and $rd_2 \geq d_1$.
If there exists a parallel tester $\{T_i\}_i$ that uses $n$ queries to an unknown isometry channel $\mathcal{V}\in\isochannel_{d_1,rd_2}$ and outputs a classical outcome $i$ with probability $P_i(\mathcal{V})$,
then there exists a parallel tester $\{\widetilde{T}_i\}_i\cup\{\widetilde{T}_\bot\}$, where $\bot$ is an extra irrelevant label outside the range of $i$, that uses $n$ queries to an unknown channel $\mathcal{E}\in\qchannel_{d_1,d_2}^r$, and outputs a classical outcome $i$ with probability $\E_{\mathcal{V}\sim \dilation_r(\mathcal{E})}[P_i(\mathcal{V})]$.
\end{theorem}

We introduce the following notation.
\begin{notation}[Estimation tasks of quantum channels]
An estimation task of quantum channels in $\qchannel_{d_1,d_2}^r$ is a set $\{A_{\mathcal{E}}\}_{\mathcal{E}\in\qchannel_{d_1,d_2}^r}$, where $A_{\mathcal{E}}$ is a set of classical outcomes that are regarded as correct answers when the unknown channel is $\mathcal{E}$.
\end{notation}
As a direct application of \cref{thm-1221159}, we have the following result.
\begin{theorem}\label{coro-12122150}
Let $d_1,d_2,r$ be positive integers, $rd_2 \geq d_1$ and $\{A_{\mathcal{E}}\}_{\mathcal{E}\in\qchannel_{d_1,d_2}^r}$ be an estimation task.
If there exists a parallel tester that uses $n$ queries to an arbitrary dilation $\mathcal{V}\in \dilation_r(\mathcal{E})$ of an unknown channel $\mathcal{E}\in\qchannel_{d_1,d_2}^r$ and outputs an $i\in A_{\mathcal{E}}$ with probability at least $1-\delta$, then there exists a parallel tester that uses $n$ queries to $\mathcal{E}$ and outputs an $i\in A_{\mathcal{E}}$ with probability at least $1-\delta$.
\end{theorem}
\begin{proof}
Let $P_i(\mathcal{V})$ be the probability of the parallel tester outputting $i$ conditioned on making queries to $\mathcal{V}$.
Therefore, we have $\sum_{i\in A_{\mathcal{E}}}P_i(\mathcal{V})\geq 1-\delta$ for any $\mathcal{V}\in\dilation_r(\mathcal{E})$.
By \cref{thm-1221159}, there exists a parallel tester that uses $n$ queries to $\mathcal{E}$ and the probability of outputting $i$ is
\[\widetilde{P}_i(\mathcal{E})=\E_{\mathcal{V}\sim \dilation_r(\mathcal{E})}[P_i(\mathcal{V})].\]
Thus, the probability of outputting an $i\in A_{\mathcal{E}}$ is
\[\sum_{i\in A_{\mathcal{E}}}\widetilde{P}_i(\mathcal{E})=\E_{\mathcal{V}\sim \dilation_r(\mathcal{E})}\left[\sum_{i\in A_\mathcal{E}} P_i(\mathcal{V})\right]\geq 1-\delta.\]
\end{proof}

\subsection{Channel tomography and estimation}
Using \cref{coro-12122150}, we can obtain the following results.
\begin{corollary}\label{coro-12120134}
There exists a parallel tester that uses $O(rd_1d_2/\varepsilon^2)$ queries to an unknown channel $\mathcal{E}\in\qchannel_{d_1,d_2}^r$ and outputs an estimate $\mathcal{F}$ such that $\|\mathcal{F}-\mathcal{E}\|_{\diamond}\leq \varepsilon$ with probability at least $2/3$, where $\|\cdot\|_{\diamond}$ is the diamond norm.
\end{corollary}


\begin{proof}

First, we define the estimation task $\{A_{\mathcal{E}}\}_{\mathcal{E}\in\qchannel_{d_1,d_2}^r}$ as
\[A_{\mathcal{E}}=\left\{\mathcal{F}\in\qchannel_{d_1,d_2}^r \,\,\big|\,\, \|\mathcal{F}-\mathcal{E}\|_\diamond \leq \varepsilon\right\}.\]
Note that any dilation in $\dilation_r(\mathcal{E})$ for $\mathcal{E}\in\qchannel_{d_1,d_2}^r$ is an isometry channel in $\isochannel_{d_1,rd_2}$.
Then, by \cref{thm-1240018}, we have a parallel tester that uses $n=O(rd_1d_2/\varepsilon^2)$ queries to a dilation $\mathcal{V}\in \dilation_r(\mathcal{E})$ and outputs $\mathcal{W}$ such that with probability at least $2/3$, we have 
\[\|\contract_r(\mathcal{W})-\mathcal{E}\|_\diamond \leq \|\mathcal{W}-\mathcal{V}\|_\diamond\leq \varepsilon,\]
where the first inequality is due to the contractivity of the diamond norm.
Let the tester output $\contract_r(\mathcal{W})$ upon getting $\mathcal{W}$. Then it can solve the task $\{A_{\mathcal{E}}\}_{\mathcal{E}\in\qchannel_{d_1,d_2}^r}$ using $n$ queries to an arbitrary dilation of $\mathcal{E}$. 
Then, by \cref{coro-12122150}, there exists a parallel tester that can also solve this task using $n$ queries to $\mathcal{E}$.

\end{proof}

Our main result can also provide the Heisenberg scaling $O(1/\varepsilon)$ for (non-unitary) quantum channel estimation tasks. 
The first example is the average-case distance tomography of quantum channels.
\begin{corollary}\label{coro-12120135}
Let $rd_2=d_1$.
There exists a parallel tester that uses $O(d_1^2/\varepsilon)$ queries to a quantum channel $\mathcal{E}\in\qchannel_{d_1,d_2}^r$ and outputs an estimate $\mathcal{F}$ such that $\|\frac{1}{d}C_\mathcal{F}-\frac{1}{d}C_\mathcal{E}\|_{1}\leq \varepsilon$ with probability at least $2/3$, where $C_{\mathcal{E}}$ denotes the unnormalized Choi operator of $\mathcal{E}$ and $\|\cdot\|_1$ is the trace norm.

There also exists a parallel tester that uses $O(\min\{d_1^{2.5}/\varepsilon,d_1^2/\varepsilon^2\})$ queries to a quantum channel $\mathcal{E}\in\qchannel_{d_1,d_2}^r$ and outputs an estimate $\mathcal{F}$ such that $\|\mathcal{F}-\mathcal{E}\|_{\diamond}\leq \varepsilon$ with probability at least $2/3$, where $\|\cdot\|_\diamond$ is the diamond norm.
\end{corollary}
\begin{proof}
We define the task $\{A_{\mathcal{E}}\}_{\mathcal{E}\in\qchannel_{d_1,d_2}^r}$ as
\[A_{\mathcal{E}}=\left\{\mathcal{F}\in\qchannel_{d_1,d_2}^r\,\,\bigg|\,\, \left\|\frac{1}{d}C_{\mathcal{F}}-\frac{1}{d}C_{\mathcal{E}}\right\|_1\leq\varepsilon\right\}.\]
Note that any dilation in $\dilation_r(\mathcal{E})$ for $\mathcal{E}\in\qchannel_{d_1,d_2}^r$ is a unitary channel in $\isochannel_{d_1,d_1}$.
By Yang-Renner-Chiribella algorithm~\cite{yang2020optimal}, we have a parallel tester that uses $n=O(d_1^2/\varepsilon)$ queries to a unitary dilation $\mathcal{U}\in\dilation_r(\mathcal{E})$ and outputs $\mathcal{W}$ such that with probability at least $2/3$, we have 
\[\left\|\frac{1}{d}C_{\contract_r(\mathcal{W})}-\frac{1}{d}C_{\mathcal{E}}\right\|_1\leq \left\|\frac{1}{d}C_\mathcal{W}-\frac{1}{d}C_{\mathcal{U}}\right\|_1=\sqrt{1-\mathrm{F}_{\mathrm{ent}}(\mathcal{W},\mathcal{U})}\leq \varepsilon,\]
where the first inequality is by the contractivity of trace norm, and the last inequality is because the Yang-Renner-Chiribella algorithm will output an estimate $\mathcal{W}$ with entanglement fidelity $\mathrm{F}_{\mathrm{ent}}(\mathcal{W},\mathcal{U})\geq 1-\varepsilon^2$, and with probability at least $2/3$.
Let the tester output $\contract_r(\mathcal{W})$ upon getting $\mathcal{W}$. Then it can solve the task $\{A_{\mathcal{E}}\}_{\mathcal{E}\in\qchannel_{d_1,d_2}^r}$ using $n$ queries to an arbitrary dilation of $\mathcal{E}$. 
Then, by \cref{coro-12122150}, there exists a parallel tester that can also solve this task using $n$ queries to $\mathcal{E}$.

Similarly, we define the task $\{A_{\mathcal{E}}\}_{\mathcal{E}\in\qchannel_{d_1,d_2}^r}$ as
\[A_{\mathcal{E}}=\left\{\mathcal{F}\in\qchannel_{d_1,d_2}^r\,\,\big|\,\, \left\|\mathcal{F}-\mathcal{E}\right\|_\diamond\leq\varepsilon\right\}.\]
By Yang-Renner-Chiribella algorithm~\cite{yang2020optimal}, we have a parallel tester that uses $n=O(d_1^2/(\varepsilon/\sqrt{d}))$ queries to a unitary dilation $\mathcal{U}\in\dilation_r(\mathcal{E})$ and outputs $\mathcal{W}$ such that with probability at least $2/3$, we have 
\[\|\contract_r(\mathcal{W})-\mathcal{E}\|_\diamond\leq \|\mathcal{W}-\mathcal{U}\|_\diamond\leq \sqrt{2d}\sqrt{1-\mathrm{F}_{\mathrm{ent}}(\mathcal{W},\mathcal{U})}\leq \varepsilon,\]
where the second inequality is due to \cite[Proposition 1.9]{haah2023query}.
Let the tester output $\contract_r(\mathcal{W})$ upon getting $\mathcal{W}$. 
Then it can solve the task $\{A_{\mathcal{E}}\}_{\mathcal{E}\in\qchannel_{d_1,d_2}^r}$ using $n$ queries to an arbitrary dilation of $\mathcal{E}$. 
Then, by \cref{coro-12122150}, there exists a parallel tester that can also solve this task using $n$ queries to $\mathcal{E}$. Then, combining this result with \cref{coro-12120134}, we know that $O(\min\{d_1^{2.5}/\varepsilon,d_1^2/\varepsilon^2\})$ queries suffice.
\end{proof}

Another example is the (mixed) state tomography using state-preparation channels.
\begin{corollary}\label{coro-12120136}
Let $rd_2=d_1$. There exists a parallel tester that uses $O(\min\{d_1^{1.5}/\varepsilon,d_1/\varepsilon^2\})$ queries to a quantum channel $\mathcal{E}\in\qchannel_{d_1,d_2}^r$ and outputs an estimate $\mathcal{F}(\ketbra{0}{0})$ such that $\|\mathcal{F}(\ketbra{0}{0})-\mathcal{E}(\ketbra{0}{0})\|_1\leq \varepsilon$ with probability at least $2/3$, where $\|\cdot\|_1$ is the trace norm.
\end{corollary}
\begin{proof}
We define the task $\{A_{\mathcal{E}}\}_{\mathcal{E}\in\qchannel_{d_1,d_2}^r}$ as
\[A_{\mathcal{E}}=\left\{\rho\in \mathcal{L}(\mathbb{C}^{d_2}) \,\,\big|\,\, \left\|\rho-\mathcal{E}(\ketbra{0}{0})\right\|_1\leq\varepsilon\right\}.\]
Note that any dilation in $\dilation_r(\mathcal{E})$ for $\mathcal{E}\in\qchannel_{d_1,d_2}^r$ is a unitary channel in $\isochannel_{d_1,d_1}$.
By Chen's algorithm~\cite{chen2025inverse}, we have a parallel tester that uses $n=O(\min\{d_1^{1.5}/\varepsilon,d_1/\varepsilon^2\})$ queries to a unitary dilation $\mathcal{U}\in\dilation_r(\mathcal{E})$ and outputs $\ket{\psi}$ such that with probability at least $2/3$, we have 
\[\|\tr_{r}(\ketbra{\psi}{\psi})-\mathcal{E}(\ketbra{0}{0})\|_1\leq \|\ketbra{\psi}{\psi}-U\ketbra{0}{0}U^\dag\|_1\leq \varepsilon,\]
where the second inequality is because Chen's algorithm will output an estimate $\ket{\psi}$ for $U\ket{0}$~\footnote{In \cite{chen2025inverse}, the author considered estimating $U\ket{d}$ for notation convenience, here we consider estimating $U\ket{0}$.} to within trace norm error $\varepsilon$ with probability at least $2/3$.
Let the tester output $\tr_r(\ketbra{\psi}{\psi})$ upon getting $\ket{\psi}$.
Then it can solve the task $\{A_{\mathcal{E}}\}_{\mathcal{E}\in\qchannel_{d_1,d_2}^r}$ using $n$ queries to an arbitrary dilation of $\mathcal{E}$. 
Then, by \cref{coro-12122150}, there exists a parallel tester that can also solve this task using $n$ queries to $\mathcal{E}$.
\end{proof}

\subsection{Construction of the local testers}
Here, we prove \cref{thm-1221159}. Our proof follows a similar idea to the construction of local testers for quantum states in \cite{chen2024local}, while generalizing it to testers for quantum channels. 

First, we define some notation that will be used in this section.
\begin{notation}
For $i\in [n]$, let $\mathcal{H}_{\mathrm{A},i}\cong\mathbb{C}^{d_1}$ and $\mathcal{H}_{\mathrm{B},i}\otimes\mathcal{H}_{\mathrm{anc},i}\cong\mathbb{C}^{d_2}\otimes\mathbb{C}^{r}$ label the input and output subsystems of the $i$-th query to the unknown isometry channel $\mathcal{V}\in\isochannel_{d_1,rd_2}$.
We have the following decompositions by Schur-Weyl duality
\[\bigotimes_{i=1}^n \mathcal{H}_{\mathrm{A},i}\otimes\mathcal{H}_{\mathrm{B},i}\stackrel{\mathfrak{S}_n \times \mathbb{U}_{d_1d_2}}{\cong} \bigoplus_{\lambda\vdash_{d_1d_2} n}\mathcal{P}_{\lambda}\otimes\mathcal{Q}^{d_1d_2}_{\lambda},\]
and
\[\bigotimes_{i=1}^n\mathcal{H}_{\mathrm{anc},i} \stackrel{\mathfrak{S}_n\times \mathbb{U}_{r}}{\cong}\bigoplus_{\lambda\vdash_{r} n}\mathcal{P}_\lambda\otimes\mathcal{Q}^r_\lambda.\]
Therefore, we have
\[\bigotimes_{i=1}^n \mathcal{H}_{\mathrm{A},i}\otimes\mathcal{H}_{\mathrm{B},i}\otimes\mathcal{H}_{\mathrm{anc},i}\stackrel{\mathfrak{S}_n\times\mathfrak{S}_n\times\mathbb{U}_{d_1d_2}\times\mathbb{U}_r}{\cong}\bigoplus_{\substack{\lambda\vdash_{d_1d_2}n\\ \mu\vdash_r n}}\mathcal{P}_{\mathrm{AB},\lambda}\otimes\mathcal{P}_{\mathrm{anc},\mu}\otimes\mathcal{Q}^{d_1d_2}_{\mathrm{AB},\lambda}\otimes\mathcal{Q}^r_{\mathrm{anc},\mu},\]
where $\mathcal{P}_{\mathrm{AB},\lambda}\otimes \mathcal{Q}^{d_1d_2}_{\mathrm{AB},\lambda}$ denotes the subspace $\mathcal{P}_{\lambda}\otimes \mathcal{Q}^{d_1d_2}_{\lambda}$ in $\bigotimes_{i=1}^n\mathcal{H}_{\mathrm{A},i}\otimes\mathcal{H}_{\mathrm{B},i}$, and $\mathcal{P}_{\mathrm{anc},\mu}\otimes \mathcal{Q}^r_{\mathrm{anc},\mu}$ denotes the subspace $\mathcal{P}_\mu\otimes\mathcal{Q}^r_\mu$ in $\bigotimes_{i=1}^n \mathcal{H}_{\mathrm{anc},i}$.
\end{notation}

Then, we provide the proof of \cref{thm-1221159}.

\begin{proof}[Proof of \cref{thm-1221159}]
Let $s\coloneqq \min\{r,d_1d_2\}$. Note that here we do not assume $r\leq d_1d_2$ (though the Kraus rank of a channel $\mathcal{E}\in\qchannel_{d_1,d_2}^{r}$ is at most $d_1d_2$).
Our construction of the tester $\{\widetilde{T}_i\}_i\cup\{\widetilde{T}_\bot\}$ is as follows.
\begin{itemize}
    \item We first construct a new tester $\{\overline{T}_i\}_i$ where
    \begin{equation}\label{eq-1242355}
    \overline{T}_i\coloneqq \E_{U\sim\mathbb{U}_{r}}[U^{\otimes n} T_iU^{\dag\otimes n}],
    \end{equation}
    where $U^{\otimes n}$ acts on $\bigotimes_{j=1}^n \mathcal{H}_{\mathrm{anc},j}$.
    
    \item Then, we define
    \begin{equation}\label{eq-1242217}
    \widetilde{T}_i\coloneqq \bigoplus_{\lambda\vdash_{s} n} \frac{1}{\dim(\mathcal{P}_\lambda)\dim(\mathcal{Q}^r_\lambda)} \cdot I_{\mathcal{P}_{\mathrm{AB},\lambda}}\otimes \tr_{\mathcal{Q}^r_{\mathrm{anc},\lambda}}\Big(\bbra{I_{\mathcal{P}_\lambda}}\overline{T}_i\kett{I_{\mathcal{P}_\lambda}}\Big),
    \end{equation}
    where $\kett{I_{\mathcal{P}_\lambda}}\in\mathcal{P}_{\mathrm{AB},\lambda}\otimes\mathcal{P}_{\mathrm{anc},\lambda}$ is the unnormalized maximally entangled state defined w.r.t. the Young's orthogonal basis (also called Young-Yamanouchi basis, on which $\pi\in\mathfrak{S}_n$ acts as a real matrix~\cite{ceccherini2010representation}).
    Note that $\widetilde{T}_i$ is a linear operator on $\bigotimes_{j=1}^n \mathcal{H}_{\mathrm{A},j}\otimes \mathcal{H}_{\mathrm{B},j}$.
\end{itemize}
To verify our construction, we first show in \cref{lemma-1242352} that $\{\overline{T}_i\}_i$ is a parallel tester that uses $n$ queries to an isometry channel $\mathcal{V}\in\isochannel_{d_1,rd_2}$ and outputs $i$ with probability 
\begin{equation*}
    \E_{\mathcal{W}\sim \dilation_r(\contract_r(\mathcal{V}))} [T_i\star C_\mathcal{W}^{\otimes n}],
\end{equation*}
and also provides an explicit expression of this probability. 
Then, using \cref{lemma-1242352}, we show in \cref{lemma-1242353} that there exists a positive semidefinite operator $\widetilde{T}_\bot$ such that $\{\widetilde{T}_i\}_i\cup\{\widetilde{T}_\bot\}$ is a parallel tester that uses $n$ queries to a quantum channel $\mathcal{E}\in\qchannel_{d_1,d_2}^r$ and outputs $i$ with probability $\E_{\mathcal{W}\sim \dilation_r(\mathcal{E})} [T_i\star C_\mathcal{W}^{\otimes n}]$, as desired. 
\end{proof}

First, we prove the following lemma about the properties of $\{\overline{T_i}\}_i$.
\begin{lemma}\label{lemma-1242352}
The tester $\{\overline{T}_i\}_i$ as defined in \cref{eq-1242355} has the following properties:
\begin{enumerate}
\item $\{\overline{T}_i\}_i$ is a parallel tester, and for any $\mathcal{V}\in\isochannel_{d_1,rd_2}$, it satisfies
\[\overline{T}_i\star C_{\mathcal{V}}^{\otimes n}=\E_{\mathcal{W}\sim\dilation_r(\contract_r(\mathcal{V}))}[T_i\star C_{\mathcal{W}}^{\otimes n}].\]

\item The probability can also be written as
\[\overline{T}_i\star C_{\mathcal{V}}^{\otimes n}=\sum_{\lambda\vdash_s n}\frac{1}{\dim(\mathcal{Q}^r_\lambda)}\tr\left(\tr_{\mathcal{Q}^r_{\mathrm{anc},\lambda}}\Big(\bbra{I_{\mathcal{P}_\lambda}}\overline{T}_i^{\mathrm{T}}\kett{I_{\mathcal{P}_\lambda}}\Big)\cdot \tr_{\mathcal{Q}^r_{\mathrm{anc},\lambda}}\Big(\ketbra{V_\lambda}{V_\lambda}\Big)\right),\]
where $s=\min\{r,d_1d_2\}$, and $\ket{V_\lambda}\in\mathcal{Q}^{d_1d_2}_{\mathrm{AB},\lambda}\otimes\mathcal{Q}^r_{\mathrm{anc},\lambda}$ is the vector appearing in the decomposition $\kett{V}^{\otimes n}=\bigoplus_{\lambda\vdash_{s}n}\kett{I_{\mathcal{P}_\lambda}}\otimes \ket{V_\lambda}$ due to \cref{lemma-3262205}.
\end{enumerate}
\end{lemma}
\begin{proof}
\textbf{Item 1}. Note that 
\begin{equation}\label{eq-1260054}
\sum_{i}\overline{T}_i=\E_{U\sim\mathbb{U}_r}\left[U^{\otimes n}\sum_i T_iU^{\dag\otimes n}\right]=\E_{U\sim\mathbb{U}_r}\left[U^{\otimes n}(\rho_\mathrm{A}\otimes I_{\mathrm{B},\mathrm{anc}}) U^{\dag\otimes n}\right]=\rho_\mathrm{A}\otimes I_{\mathrm{B},\mathrm{anc}},
\end{equation}
where $\rho_\mathrm{A}$ is a density operator on $\bigotimes_{j=1}^n\mathcal{H}_{\mathrm{A},j}$, $I_{\mathrm{B},\mathrm{anc}}$ is the identity operator on $\bigotimes_{j=1}^n\mathcal{H}_{\mathrm{B},j}\otimes\mathcal{H}_{\mathrm{anc},j}$, and we use the fact that $\{T_i\}_i$ is a parallel tester and $U^{\otimes n}$ acts only on $\bigotimes_{j=1}^n \mathcal{H}_{\mathrm{anc},j}$.
Therefore, $\{\overline{T}_i\}_i$ is a parallel tester.
On the other hand, note that
\begin{align}
\overline{T}_i\star C_{\mathcal{V}}^{\otimes n}&=\tr\!\left(\overline{T}_i^{\mathrm{T}} C_{\mathcal{V}}^{\otimes n}\right) \nonumber \\
&= \E_{U\sim \mathbb{U}_r} \left[\tr\!\left(T_i^{\mathrm{T}}U^{\otimes n} C_\mathcal{V}^{\otimes n} U^{\dag\otimes n}\right)\right]\label{eq:use-def-Tbar} \\
&=\E_{U\sim \mathbb{U}_r} \left[\tr\!\left(T_i^{\mathrm{T}}  C_{\mathcal{U}\circ \mathcal{V}}^{\otimes n}\right)\right]\nonumber\\
&=\E_{\mathcal{W}\sim\dilation_r(\contract_r(\mathcal{V}))}\left[\tr\!\left(T_i^\mathrm{T} C_{\mathcal{W}}^{\otimes n}\right)\right]\label{eq-1252216}\\
&=\E_{\mathcal{W}\sim\dilation_r(\contract_r(\mathcal{V}))}\left[T_i \star C_{\mathcal{W}}^{\otimes n}\right] \nonumber
\end{align}
where \Cref{eq:use-def-Tbar} uses the definition of $\overline{T}_i$ in \Cref{eq-1242355},
and \cref{eq-1252216} is due to the unitary freedom of Stinespring dilation and the definition of the Haar distribution on $\dilation_r(\cdot)$. 

\textbf{Item 2}. We consider $\kett{V}$ as a bipartite state in $(\mathbb{C}^{d_1}\otimes \mathbb{C}^{d_2})\otimes \mathbb{C}^{r}$. Then, by \cref{lemma-3262205}, we can write $\kett{V}^{\otimes n}=\bigoplus_{\lambda\vdash_s n}\kett{I_{\mathcal{P}_\lambda}}\otimes\ket{V_\lambda}$, where $\kett{I_{\mathcal{P}_\lambda}}\in\mathcal{P}_{\mathrm{AB},\lambda}\otimes\mathcal{P}_{\mathrm{anc},\lambda}$ is an unnormalized maximally entangled state and $\ket{V_\lambda}\in \mathcal{Q}^{d_1d_2}_{\mathrm{AB},\lambda}\otimes\mathcal{Q}^r_{\mathrm{anc},\lambda}$. Also note that $U^{\otimes n}\overline{T}_i = \overline{T}_i U^{\otimes n}$ for any $U\in\mathbb{U}_r$ where $U^{\otimes n}$ acts on $\bigotimes_{j=1}^n \mathcal{H}_{\mathrm{anc},j}$.
Therefore,
\begin{align}
\overline{T}_i\star C_{\mathcal{V}}^{\otimes n}&=\tr\!\left(\overline{T}_i^\mathrm{T} \kettbbra{V}{V}^{\otimes n}\right)\nonumber \\
&=\tr\!\left(\overline{T}_i^\mathrm{T} \E_{U\sim\mathbb{U}_r}[U^{\otimes n}\kettbbra{V}{V}^{\otimes n} U^{\dag\otimes n}]\right) \nonumber  \\
&=\tr\!\left(\overline{T}_i^\mathrm{T} \E_{U\sim\mathbb{U}_r}\left[\bigoplus_{\lambda,\mu\vdash_s n} \kett{I_{\mathcal{P}_\lambda}}\bbra{I_{\mathcal{P}_\mu}}\otimes \texttt{q}_\lambda(U)\ketbra{V_\lambda}{V_\mu}\texttt{q}_\mu(U)^\dag \right]\right)\label{eq-1260033}\\
&=\tr\!\left(\overline{T}_i^{\mathrm{T}} \cdot \left(\bigoplus_{\lambda\vdash_s n}\kett{I_{\mathcal{P}_\lambda}}\bbra{I_{\mathcal{P}_\lambda}}\otimes \tr_{\mathcal{Q}^r_{\mathrm{anc},\lambda}}\Big(\ketbra{V_\lambda}{V_\lambda}\Big)\otimes \frac{1}{\dim(\mathcal{Q}^r_\lambda)}I_{\mathcal{Q}^r_{\mathrm{anc},\lambda}}\right)\right)\label{eq-1260034}\\
&=\sum_{\lambda\vdash_s n} \frac{1}{\dim(\mathcal{Q}^r_\lambda)}\tr\left(\tr_{\mathcal{Q}^r_{\mathrm{anc},\lambda}}\Big(\bbra{I_{\mathcal{P}_\lambda}}\overline{T}_i^{\mathrm{T}} \kett{I_{\mathcal{P}_\lambda}}\Big)\cdot \tr_{\mathcal{Q}^r_{\mathrm{anc},\lambda}}\Big(\ketbra{V_\lambda}{V_\lambda}\Big)\right),\nonumber
\end{align}
where in \cref{eq-1260033} $\texttt{q}_\lambda(U)$ acts on $\mathcal{Q}^r_{\mathrm{anc},\lambda}$, and \cref{eq-1260034} is by using Schur's lemma~\cite{fulton2013representation}.
\end{proof}

Then, we prove the following lemma about the properties of $\{\widetilde{T}_i\}_i$.

\begin{lemma}\label{lemma-1242353}
The operators $\{\widetilde{T}_i\}_i$ as defined in \cref{eq-1242217} have the following properties:
\begin{enumerate}
\item There exists a positive semidefinite operator $\widetilde{T}_\bot$ such that $\{\widetilde{T}_i\}_i\cup\{\widetilde{T}_\bot\}$ is a parallel tester. 
\item For any $\mathcal{E}\in\qchannel_{d_1,d_2}^r$, we have
\[\widetilde{T}_i\star C_{\mathcal{E}}^{\otimes n}=\E_{\mathcal{W}\sim\dilation_r(\mathcal{E})}[T_i\star C_{\mathcal{W}}^{\otimes n}].\]
\end{enumerate}
\end{lemma}
\begin{proof}
\textbf{Item 1}.
Note that by the definition in \Cref{eq-1242217},
\begin{align}
\sum_i \widetilde{T}_i&=\bigoplus_{\lambda\vdash_s n}\frac{1}{\dim(\mathcal{P}_\lambda)\dim(\mathcal{Q}^r_\lambda)}\cdot I_{\mathcal{P}_{\mathrm{AB},\lambda}}\otimes \tr_{\mathcal{Q}^r_{\mathrm{anc},\lambda}}\left(\bbra{I_{\mathcal{P}_\lambda}}\sum_i \overline{T}_i\kett{I_{\mathcal{P}_\lambda}}\right) \nonumber\\
&=\bigoplus_{\lambda\vdash_s n}\frac{1}{\dim(\mathcal{P}_\lambda)\dim(\mathcal{Q}^r_\lambda)}\cdot I_{\mathcal{P}_{\mathrm{AB},\lambda}}\otimes \tr_{\mathcal{Q}^r_{\mathrm{anc},\lambda}}\Big(\bbra{I_{\mathcal{P}_\lambda}} \rho_\mathrm{A}\otimes I_{\mathrm{B},\mathrm{anc}} \kett{I_{\mathcal{P}_\lambda}}\Big),\label{eq-1260117}
\end{align}
where in \cref{eq-1260117} $\rho_\mathrm{A}$ is a density operator on $\bigotimes_{j=1}^n\mathcal{H}_{\mathrm{A},j}$, $I_{\mathrm{B},\mathrm{anc}}$ is the identity operator on $\bigotimes_{j=1}^n\mathcal{H}_{\mathrm{B},j}\otimes\mathcal{H}_{\mathrm{anc},j}$ and we note that $U^{\otimes n}$ acts on $\bigotimes_{j=1}^n \mathcal{H}_{\mathrm{anc},j}$, as shown in \cref{eq-1260054}.
Now, we write $\rho_\mathrm{A}\otimes I_{\mathrm{B},\mathrm{anc}}=(\rho_\mathrm{A}\otimes I_{\mathrm{B}})\otimes I_{\mathrm{anc}}$ and decompose $\rho_\mathrm{A}\otimes I_{\mathrm{B}}$ in the Schur-Weyl basis:
\[\rho_\mathrm{A}\otimes I_{\mathrm{B}}=\bigoplus_{\substack{\lambda,\mu\vdash_{d_1d_2}n}} M_{\lambda\rightarrow \mu},\]
where $M_{\lambda\rightarrow\mu}$ is a linear operator from $\mathcal{P}_{\mathrm{AB},\lambda}\otimes\mathcal{Q}^{d_1d_2}_{\mathrm{AB},\lambda}$ to $\mathcal{P}_{\mathrm{AB},\mu}\otimes\mathcal{Q}^{d_1d_2}_{\mathrm{AB},\mu}$, and since $\rho_\mathrm{A}\otimes I_{\mathrm{B}}$ is positive semidefinite, $M_{\lambda\rightarrow\lambda}$ is also positive semidefinite. 
Furthermore, we can write
\[(\rho_\mathrm{A}\otimes I_{\mathrm{B}})\otimes I_{\mathrm{anc}}=\bigoplus_{\substack{\lambda,\mu\vdash_{d_1d_2}n \\ \nu\vdash_r n}}M_{\lambda\rightarrow\mu}\otimes I_{\mathcal{P}_{\mathrm{anc},\nu}}\otimes I_{\mathcal{Q}^r_{\mathrm{anc},\nu}}.\]
Then, \cref{eq-1260117} equals
\begin{align}
&\bigoplus_{\lambda\vdash_s n}\frac{1}{\dim(\mathcal{P}_\lambda)\dim(\mathcal{Q}^r_\lambda)}\cdot I_{\mathcal{P}_{\mathrm{AB},\lambda}}\otimes \tr_{\mathcal{Q}^r_{\mathrm{anc},\lambda}}\left(\bbra{I_{\mathcal{P}_\lambda}}\bigoplus_{\substack{\kappa,\mu\vdash_{d_1d_2}n \\ \nu\vdash_r n}}M_{\kappa\rightarrow\mu}\otimes I_{\mathcal{P}_{\mathrm{anc},\nu}}\otimes I_{\mathcal{Q}^r_{\mathrm{anc},\nu}}\kett{I_{\mathcal{P}_\lambda}}\right) \nonumber \\
=&\bigoplus_{\lambda\vdash_s n}\frac{1}{\dim(\mathcal{P}_\lambda)\dim(\mathcal{Q}^r_\lambda)}\cdot I_{\mathcal{P}_{\mathrm{AB},\lambda}}\otimes \tr_{\mathcal{Q}^r_{\mathrm{anc},\lambda}}\Big(\tr_{\mathcal{P}_{\mathrm{AB},\lambda}}(M_{\lambda\rightarrow\lambda})\otimes I_{\mathcal{Q}^r_{\mathrm{anc},\lambda}}\Big) \label{eq-1261541}\\
=& \bigoplus_{\lambda\vdash_s n}\frac{1}{\dim(\mathcal{P}_\lambda)}\cdot I_{\mathcal{P}_{\mathrm{AB},\lambda}}\otimes \tr_{\mathcal{P}_{\mathrm{AB},\lambda}}(M_{\lambda\rightarrow\lambda})\nonumber \\
\sqsubseteq & \bigoplus_{\lambda\vdash_{d_1d_2} n}\frac{1}{\dim(\mathcal{P}_\lambda)}\cdot I_{\mathcal{P}_{\mathrm{AB},\lambda}}\otimes \tr_{\mathcal{P}_{\mathrm{AB},\lambda}}(M_{\lambda\rightarrow\lambda}),\label{eq-1261542}
\end{align}
where \cref{eq-1261541} is because $\kett{I_{\mathcal{P}_\lambda}}$ is an unnormalized maximally entangled state on $\mathcal{P}_{\mathrm{AB},\lambda}\otimes\mathcal{P}_{\mathrm{anc},\lambda}$, \cref{eq-1261542} is because $M_{\lambda\rightarrow\lambda}$ is positive semidefinite and $s\leq d_1d_2$.
Then, note that
\begin{align}
\frac{1}{n!}\sum_{\pi\in\mathfrak{S}_n}\texttt{p}_{\mathrm{A}}(\pi)\,\rho_\mathrm{A}\, \texttt{p}_{\mathrm{A}}(\pi)^\dag \otimes I_{\mathrm{B}}&=\frac{1}{n!}\sum_{\pi\in \mathfrak{S}_n}\texttt{p}_{\mathrm{AB}}(\pi)(\rho_\mathrm{A}\otimes I_{\mathrm{B}})\texttt{p}_{\mathrm{AB}}(\pi)^\dag \nonumber \\
&=\bigoplus_{\lambda\vdash_{d_1d_2} n}\frac{1}{\dim(\mathcal{P}_\lambda)} \cdot I_{\mathcal{P}_{\mathrm{AB},\lambda}}\otimes \tr_{\mathcal{P}_{\mathrm{AB},\lambda}}(M_{\lambda\rightarrow\lambda}), \label{eq-1261722}
\end{align}
where $\texttt{p}_\mathrm{A}(\pi)$ and $\texttt{p}_{\mathrm{AB}}(\pi)$ are the permutation actions of $\pi$ on $\bigotimes_{j=1}^n \mathcal{H}_{\mathrm{A},j}$ and $\bigotimes_{j=1}^n \mathcal{H}_{\mathrm{A},j}\otimes\mathcal{H}_{\mathrm{B},j}$, respectively, and \cref{eq-1261722} is due to Schur's lemma.
Note that \cref{eq-1261722} is exactly equal to \cref{eq-1261542}. Therefore, we have
\[\sum_{i} \widetilde{T}_i\sqsubseteq \rho'_\mathrm{A} \otimes I_{\mathrm{B}},\]
where $\rho'_\mathrm{A}=\frac{1}{n!}\sum_{\pi\in\mathfrak{S}_n} \texttt{p}_\mathrm{A}(\pi) \,\rho_\mathrm{A}\, \texttt{p}_\mathrm{A}(\pi)^\dag$ is a quantum state. This means that we can find a positive semidefinite operator $\widetilde{T}_\bot$ such that $\sum_i\widetilde{T}_i+\widetilde{T}_\bot=\rho'_\mathrm{A}\otimes I_{\mathrm{B}}$, and thus $\{\widetilde{T}_i\}_i\cup\{\widetilde{T}_\bot\}$ is a parallel tester.

\textbf{Item 2}. Let $\mathcal{E}\in\qchannel_{d_1,d_2}^r$. We choose an arbitrary dilation $\mathcal{V}\in\dilation_r(\mathcal{E})$, where $V: \mathcal{H}_{\mathrm{A}} \rightarrow \mathcal{H}_{\mathrm{B}}\otimes\mathcal{H}_\mathrm{anc}$ and $\mathcal{H}_\mathrm{A}\cong \mathbb{C}^{d_1}$, $\mathcal{H}_{\mathrm{B}}\cong\mathbb{C}^{d_2}$, $\mathcal{H}_{\mathrm{anc}}\cong \mathbb{C}^{r}$. 
Note that
\[\tr_{\mathcal{H}_\mathrm{anc}}(C_{\mathcal{V}})=\tr_{\mathcal{H}_\mathrm{anc}}(\kettbbra{V}{V})=C_\mathcal{E}.\]
Considering $\kett{V}$ as a bipartite state in $(\mathcal{H}_\mathrm{A}\otimes\mathcal{H}_\mathrm{B})\otimes\mathcal{H}_{\mathrm{anc}}$, we can write $\kett{V}^{\otimes n}=\bigoplus_{\lambda\vdash_s n}\kett{I_{\mathcal{P}_\lambda}}\otimes \ket{V_\lambda}$ for $\kett{I_{\mathcal{P}_\lambda}}\in\mathcal{P}_{\mathrm{AB},\lambda}\otimes\mathcal{P}_{\mathrm{anc},\lambda}$ and $\ket{V_\lambda}\in\mathcal{Q}^{d_1d_2}_{\mathrm{AB},\lambda}\otimes\mathcal{Q}^r_{\mathrm{anc},\lambda}$ due to \cref{lemma-3262205}.
Therefore, we have
\begin{align}
\tr_{\mathrm{anc}}(\kettbbra{V}{V}^{\otimes n})&=\tr_{\mathrm{anc}}\left(\bigoplus_{\substack{\lambda,\mu\vdash_{s}n}} \kettbbra{I_{\mathcal{P}_\lambda}}{I_{\mathcal{P}_\mu}}\otimes \ketbra{V_\lambda}{V_\mu}\right) \nonumber\\
&=\bigoplus_{\lambda\vdash_s n} \tr_{\mathcal{P}_{\mathrm{anc},\lambda}}\Big(\kettbbra{I_{\mathcal{P}_\lambda}}{I_{\mathcal{P}_\lambda}}\Big)\otimes \tr_{\mathcal{Q}^r_{\mathrm{anc},\lambda}}\Big(\ketbra{V_\lambda}{V_\lambda}\Big)\nonumber \\
&=\bigoplus_{\lambda\vdash_s n}I_{\mathcal{P}_{\mathrm{AB},\lambda}}\otimes \tr_{\mathcal{Q}^r_{\mathrm{anc},\lambda}}\Big(\ketbra{V_\lambda}{V_\lambda}\Big),\label{eq-12151045}
\end{align}
where $\tr_{\mathrm{anc}}(\cdot)$ denotes the partial trace on $\bigotimes_{j=1}^n \mathcal{H}_{\mathrm{anc},j}$. We also have
\begin{align}
C_\mathcal{E}^{\otimes n}=\bigoplus_{\lambda\vdash_{d_1d_2} n} I_{\mathcal{P}_{\mathrm{AB},\lambda}}\otimes C_{\mathcal{E},\lambda},\label{eq-12151044}
\end{align}
for some $\mathcal{C}_{\mathcal{E},\lambda}\in\mathcal{L}(\mathcal{Q}^{d_1d_2}_{\mathrm{AB},\lambda})$. By comparing \cref{eq-12151045} with \cref{eq-12151044}, we know that $\tr_{\mathcal{Q}^r_{\mathrm{anc},\lambda}}(\ketbra{V_\lambda}{V_\lambda})=C_{\mathcal{E},\lambda}$ for $\lambda\vdash_s n$, and $C_{\mathcal{E},\lambda}=0$ for those $\lambda$ that have more than $s$ rows.
Therefore,
\begin{align}
\widetilde{T}_i\star C_{\mathcal{E}}^{\otimes n}&=\tr\left(\left(\bigoplus_{\lambda\vdash_{s} n} \frac{1}{\dim(\mathcal{P}_\lambda)\dim(\mathcal{Q}^r_\lambda)} \cdot I_{\mathcal{P}_{\mathrm{AB},\lambda}}\otimes \tr_{\mathcal{Q}^r_{\mathrm{anc},\lambda}}\Big(\bbra{I_{\mathcal{P}_\lambda}}\overline{T}_i\kett{I_{\mathcal{P}_\lambda}}\Big)\right)^{\mathrm{T}} \cdot C^{\otimes n}_\mathcal{E}\right)\nonumber \\
&=\sum_{\lambda\vdash_s n}\frac{1}{\dim(\mathcal{Q}^r_\lambda)}\tr\left(\tr_{\mathcal{Q}^r_{\mathrm{anc},\lambda}}\Big(\bbra{I_{\mathcal{P}_\lambda}}\overline{T}_i\kett{I_{\mathcal{P}_\lambda}}\Big)^{\mathrm{T}}\cdot C_{\mathcal{E},\lambda}\right) \nonumber \\
&=\sum_{\lambda\vdash_s n}\frac{1}{\dim(\mathcal{Q}^r_\lambda)}\tr\left(\tr_{\mathcal{Q}^r_{\mathrm{anc},\lambda}}\Big(\bbra{I_{\mathcal{P}_\lambda}}\overline{T}_i^{\mathrm{T}}\kett{I_{\mathcal{P}_\lambda}}\Big)\cdot \tr_{\mathcal{Q}^r_{\mathrm{anc},\lambda}}\Big(\ketbra{V_\lambda}{V_\lambda}\Big)\right) \nonumber\\
&=\overline{T}_i\star C_{\mathcal{V}}^{\otimes n} \label{eq-1270039}\\
&=\E_{\mathcal{W}\sim \dilation_r(\mathcal{E})}[T_i\star C_{\mathcal{W}}^{\otimes n}], \label{eq-1270040}
\end{align}
where \cref{eq-1270039} is by item 2 of \cref{lemma-1242352} and \cref{eq-1270040} is by item 1 of \cref{lemma-1242352}.
\end{proof}

Then, we introduce the following result about bipartite pure states, which is widely used in quantum information theory (see, e.g., \cite{matsumoto2007universal}).
\begin{lemma}\label{lemma-3262205}
Let \(\ket{\psi}\in \mathcal{H}_{\mathrm{A}}\otimes\mathcal{H}_{\mathrm{B}}\cong \mathbb{C}^{d_1}\otimes\mathbb{C}^{d_2}\) be a vector and let $s=\min\{d_1, d_2\}$, then $\ket{\psi}^{\otimes n}$ can be written as
\begin{equation*}
\ket{\psi}^{\otimes n}= \bigoplus_{\lambda\vdash_{s} n}\kett{I_{\mathcal{P}_\lambda}}\otimes\ket{\psi_\lambda},
\end{equation*}
where \(\kett{I_{\mathcal{P}_\lambda}}\) is the unnormalized maximally entangled state on \(\mathcal{P}_{\mathrm{A},\lambda}\otimes \mathcal{P}_{\mathrm{B},\lambda}\) defined w.r.t. the Young's orthogonal basis, and \(\ket{\psi_\lambda}\in\mathcal{Q}^{d_1}_{\mathrm{A},\lambda}\otimes \mathcal{Q}^{d_2}_{\mathrm{B},\lambda}\).
\end{lemma}
\begin{proof}
Note that $\ket{\psi}^{\otimes n}$ is invariant under the ``simultaneous permutation'' action $\texttt{p}_{\mathrm{A}}(\pi)\otimes\texttt{p}_{\mathrm{B}}(\pi)$ for any $\pi\in\mathfrak{S}_n$.
On the other hand, by the Schur-Weyl duality, we know that
\begin{align}
\frac{1}{n!}\sum_{\pi\in\mathfrak{S}_n} \texttt{p}_{\mathrm{A}}(\pi)\otimes\texttt{p}_{\mathrm{B}}(\pi)&=\bigoplus_{\substack{\lambda\vdash_{d_1} n\\ \mu\vdash_{d_2}n}} \texttt{p}_{\mathrm{A},\lambda}(\pi)\otimes\texttt{p}_{\mathrm{B},\mu}(\pi)\otimes I_{\mathcal{Q}_{\mathrm{A},\lambda}^{d_1}}\otimes I_{\mathcal{Q}_{\mathrm{B},\mu}^{d_2}} \nonumber\\
&=\bigoplus_{\substack{\lambda\vdash_{d_1} n\\ \mu\vdash_{d_2}n}} \texttt{p}^*_{\mathrm{A},\lambda}(\pi)\otimes\texttt{p}_{\mathrm{B},\mu}(\pi)\otimes I_{\mathcal{Q}_{\mathrm{A},\lambda}^{d_1}}\otimes I_{\mathcal{Q}_{\mathrm{B},\mu}^{d_2}}\label{eq-12131226}\\
&=\bigoplus_{\lambda\vdash_s n}\frac{1}{\dim(\mathcal{P}_\lambda)} \kettbbra{I_{\mathcal{P}_\lambda}}{I_{\mathcal{P}_\lambda}} \otimes I_{\mathcal{Q}_{\mathrm{A},\lambda}^{d_1}}\otimes I_{\mathcal{Q}_{\mathrm{B},\lambda}^{d_2}}, \label{eq-12131216}
\end{align}
where in \cref{eq-12131226} the $(\cdot)^*$ is defined w.r.t. the Young's orthogonal basis so that $\texttt{p}_{\lambda}(\pi)$ is a real matrix~\cite{ceccherini2010representation}, and \cref{eq-12131216} is because the only subspace that is invariant under $\texttt{p}^*_{\mathrm{A},\lambda}(\pi)\otimes\texttt{p}_{\mathrm{B},\mu}(\pi)$ is spanned by $\kett{I_{\mathcal{P}_\lambda}}\in\mathcal{P}_{\mathrm{A},\lambda}\otimes\mathcal{P}_{\mathrm{B},\mu}$ when $\lambda=\mu$, and zero space $\{0\}$ otherwise (this can be seen by considering the isomorphism of representations $\mathcal{P}_{\mathrm{A},\lambda}^*\otimes\mathcal{P}_{\mathrm{B},\mu}\stackrel{\mathfrak{S}_n}{\cong} \mathcal{L}(\mathcal{P}_{\mathrm{A},\lambda},\mathcal{P}_{\mathrm{B},\mu})$ and all linear operators in $\mathcal{L}(\mathcal{P}_{\mathrm{A},\lambda},\mathcal{P}_{\mathrm{B},\mu})$ that commute with the action of $\pi$ are proportional to the identity operator when $\lambda=\mu$ and $0$ otherwise, by Schur's lemma). This means $\ket{\psi}^{\otimes n}$ must be in the support of the projector $\frac{1}{n!}\sum_{\pi\in\mathfrak{S}_n}\texttt{p}_\mathrm{A}(\pi)\otimes\texttt{p}_{\mathrm{B}}(\pi)$ and thus has the form $\ket{\psi}^{\otimes n}= \bigoplus_{\lambda\vdash_{s} n}\kett{I_{\mathcal{P}_\lambda}}\otimes\ket{\psi_\lambda}$.

\end{proof}

\bibliographystyle{alpha}
\bibliography{main}

\appendix

\section{Isometry channel tomography in diamond norm}
\label{sec:isometry_channel_tomography}

In this appendix, we prove the following lemma,
which extends the $O(d^2/\varepsilon^2)$ unitary channel tomography algorithm in \cite{haah2023query}
to isometry channel tomography.

\begin{lemma}[Isometry channel tomography]\label{thm-1240018}
Let $d_1\leq d_2$ be two positive integers, $\varepsilon\in (0,1)$, $V:\mathbb{C}^{d_1}\rightarrow\mathbb{C}^{d_2}$ be an isometry and $\mathcal{V}=V(\cdot)V^\dag\in \isochannel_{d_1,d_2}$ be the corresponding isometry channel.
There exists an algorithm that uses $O(d_1d_2/\varepsilon^2)$ queries to $\mathcal{V}$ and outputs an isometry channel estimate $\widehat{\mathcal{V}}$ such that $\|\mathcal{V}-\widehat{\mathcal{V}}\|_\diamond\leq \varepsilon$ with probability $\geq 2/3$. Moreover, these queries are used in parallel.
\end{lemma}

The core ingredient to prove \Cref{thm-1240018} is the following lemma~\cite{CL14,KRT17,GKKT20,haah2023query} for pure state tomography.

\begin{lemma}[Pure state tomography, c.f.\ {\cite[Proposition 2.2]{haah2023query}}]
    \label{lmm:pure-state-tomo}
    Let $d$ be a positive integer.
    There exists a pure state tomography algorithm that uses $O(d/\varepsilon_{\max})$ copies of the input quantum state $\ket{v}\in \Co^d$ and outputs a quantum state estimate (by a classical description) $\ket{\widehat{v}}$ such that
    \begin{equation*}
        \ket{\widehat{v}}= \phi \sqrt{1-\varepsilon} \ket{v} + \sqrt{\varepsilon} \ket{w},
    \end{equation*}
    where $\phi$ is a random phase, $\varepsilon\in [0,1]$ is a random number with $\Pr[\varepsilon\leq \varepsilon_{\max}]\geq 1- \exp(-5d)$, and $\ket{w}$ is a Haar random state orthogonal to $\ket{v}$.
\end{lemma}

The second lemma we need is to convert a weak tomography algorithm into a standard tomography algorithm as required by \Cref{thm-1240018}.

\begin{lemma}
    \label{lmm:weak-tomo}
    Let $d_1\leq d_2$ be two positive integers, $V:\mathbb{C}^{d_1}\rightarrow\mathbb{C}^{d_2}$ be an isometry, and $\mathcal{V}=V(\cdot)V^\dag\in \isochannel_{d_1,d_2}$ be the corresponding isometry channel.
    Let $\calA$ be a weak isometry channel tomography algorithm such that given queries to $\calV$, it outputs an isometry estimate $\widehat{V}:\mathbb{C}^{d_1}\rightarrow\mathbb{C}^{d_2}$ such that 
    \begin{equation}
        \label{eq:weak-tomo-precond}
        \Pr\bracks*{\exists \textup{ diagonal unitary } \Phi:\Co^{d_1}\rightarrow \Co^{d_1}, \norm*{V\Phi- \widehat{V}}_{\textup{op}}\leq \varepsilon\leq \frac{1}{8}} \geq 1-\eta,
    \end{equation}
    where $\norm*{\cdot}_{\textup{op}}$ denotes the operator norm.
    Then, there exists an isometry channel tomography algorithm that uses $\calA$ twice in parallel and outputs an isometry estimate $\widehat{\calV'}$ such that
    \begin{equation*}
        \Pr\bracks*{\norm*{\calV -\widehat{\calV'}}_{\diamond}\leq 98\varepsilon}\geq 1-2\eta.
    \end{equation*}
\end{lemma}

\begin{proof}
    We extend the proof of \cite[Proposition 2.3]{haah2023query} for unitary channel tomography to the case of isometry channel tomography.
    Suppose $\calA$ is a weak isometry channel tomography algorithm as described in \Cref{lmm:weak-tomo}.
    Let us apply $\calA$ using queries to $\calV$ to obtain an isometry estimate $\widehat{V_1}:\Co^{d_1}\rightarrow \Co^{d_2}$.
    In parallel, we apply $\calA$ using queries to $\calV \circ \calF$
    to obtain another isometry estimate $\widehat{V_2}:\Co^{d_1}\rightarrow \Co^{d_2}$,
    where $\calF$ is the quantum channel for quantum Fourier transform $\mathit{F}:\Co^{d_1}\rightarrow \Co^{d_1}$.

    By our condition of $\calA$ and the union bound, we have
    \begin{equation}
        \label{eq:two-weak-tomo}
        \norm*{V\Phi_1 - \widehat{V_1}}_{\textup{op}}\leq \varepsilon\quad \textup{and}\quad \norm*{V F \Phi_2 -\widehat{V_2}}_{\textup{op}}\leq \varepsilon 
    \end{equation}
    for some diagonal unitaries $\Phi_1,\Phi_2:\Co^{d_1}\rightarrow \Co^{d_1}$,
    with probability $\geq 1-2\eta$.
    In this case, since $V$ is an isometry, we have $V^\dagger V = \sum_{j=0}^{d_1-1} \ket{j}\!\bra{j}=I_{d_1}$ and consequently
    \begin{align*}
        \norm*{\widehat{V_1}^\dagger \widehat{V_2}- \Phi_1^\dagger F \Phi_2}_{\textup{op}}
        &\leq \norm*{\parens*{\widehat{V_1}^\dagger -  \Phi_1^\dagger V^\dagger} \widehat{V_2}}_{\textup{op}} + \norm*{ \Phi_1^\dagger V^\dagger\parens*{\widehat{V_2} - VF \Phi_2}}_{\textup{op}}\\
        &\leq \norm*{V\Phi_1 - \widehat{V_1}}_{\textup{op}} + \norm*{V F \Phi_2 -\widehat{V_2}}_{\textup{op}}\\
        &\leq 2\varepsilon.
    \end{align*}
    Let $p(k,j)$ be the proposition 
    \begin{equation}
        \label{eq:def-p-k-j}
        \abs*{\bra{k}\parens*{\widehat{V_1}^\dagger \widehat{V_2}- \Phi_1^\dagger F \Phi_2}\ket{j}}\leq \frac{4\varepsilon}{\sqrt{d_1}}.
    \end{equation}
    From the pigeonhole principle, we have
    \begin{equation}
        \label{eq:condition-k-j}
        \textup{for any $j=0$ to $d_1-1$,}\quad \#\braces*{k: p(k,j)} \geq \frac{3d_1}{4}, 
    \end{equation}
    with probability $\geq 1-2\eta$, 
    where $\#\braces*{k: p(k,j)}$
    denotes the number of $k$ such that $p(k,j)$ is satisfied.
    
    Let $\Phi_3= \sum_{k,j=0}^{d_1-1}\frac{\bra{k}\widehat{V_1}^\dagger \widehat{V_2} \ket{j}}{\bra{k}F\ket{j}}\ket{k}\!\bra{j}$.
    If $p(k,j)$ is satisfied, then from \Cref{eq:def-p-k-j}, we have
    \begin{equation}
        \label{eq:phi321}
        \abs*{\bra{k}\Phi_3\ket{j} - \bra{k}\Phi_1^\dagger\ket{k}\cdot \bra{j}\Phi_2\ket{j}} \leq 4\varepsilon,
    \end{equation}
     where we use that $\abs*{\bra{k}F\ket{j}}=\frac{1}{\sqrt{d_1}}$ for any $k,j$.
    Further, if both $p(k,0)$ and $p(k,j)$ are satisfied, then \Cref{eq:phi321} implies
    \begin{equation}
        \label{eq:estimate-Phi3}
        \abs*{\frac{\bra{k}\Phi_3\ket{j}}{\bra{k}\Phi_3\ket{0}}-\frac{\bra{j}\Phi_2\ket{j}}{\bra{0}\Phi_2\ket{0}}} \leq \frac{2\cdot 4\varepsilon}{1-4\varepsilon} \leq 16 \varepsilon.
    \end{equation}
    By \Cref{eq:condition-k-j}, we have
    \begin{equation}
        \label{eq:both-pkj-pk0}
        \textup{for any $j=0$ to $d_1-1$,}\quad \#\braces*{k: p(k,j) \wedge p(k,0)} \geq \frac{d_1}{2},
    \end{equation}
    with probability $\geq 1-2\eta$.
    For each $j=0$ to $d_1-1$, let $a_j$ and $b_j$ be the medians of the real parts and the imaginary parts of the set $\braces*{\frac{\bra{k}\Phi_3\ket{j}}{\bra{k}\Phi_3\ket{0}}}_k$, respectively. 
    Then, \Cref{eq:estimate-Phi3,eq:both-pkj-pk0} lead to
    \begin{equation*}
        \abs*{\parens*{a_j + i b_j} - \frac{\bra{j}\Phi_2\ket{j}}{\bra{0}\Phi_2\ket{0}}}\leq \sqrt{(16\varepsilon)^2 + (16\varepsilon)^2},
    \end{equation*}
    which further implies that $\phi_j =\frac{a_j + i b_j}{\abs*{a_j + i b_j}}$ satisfies
    $\abs*{\phi_j-\frac{\bra{j}\Phi_2\ket{j}}{\bra{0}\Phi_2\ket{0}}}\leq 48\varepsilon$.
    Let $\Phi=\sum_{j=0}^{d_1-1} \phi_j\ket{j}\!\bra{j}$,
    then 
    \begin{equation}
        \label{eq:Phi-estimate}
        \norm*{\bra{0}\Phi_2\ket{0}\cdot\Phi - \Phi_2}_{\textup{op}}\leq 48\varepsilon
    \end{equation}
    with probability $\geq 1-2\eta$.
    Then, we have
    \begin{align}
        \norm*{\calV-\widehat{\calV_2} \Phi^\dagger \calF^\dagger}_{\diamond}&\leq 2
        \norm*{V\bra{0}\Phi_2\ket{0} -\widehat{V_2} \Phi ^\dagger F^\dagger }_{\textup{op}} \label{eq-221646}\\
        &\leq 
        2\norm*{V\bra{0}\Phi_2\ket{0}- \widehat{V_2} \Phi_2^\dagger F^\dagger \bra{0}\Phi_2\ket{0}}_{\textup{op}}
        + 2\norm*{\widehat{V_2} \Phi_2^\dagger F^\dagger \bra{0}\Phi_2\ket{0} - \widehat{V_2} \Phi^\dagger F^\dagger}_{\textup{op}} \nonumber\\
        &=
        2\norm*{V- \widehat{V_2} \Phi_2^\dagger F^\dagger}_{\textup{op}}
        +
        2\norm*{\bra{0}\Phi_2\ket{0}\cdot\Phi - \Phi_2}_{\textup{op}} \label{eq-221647}\\
        &\leq 98\varepsilon,\label{eq-221648}
    \end{align}
    with probability $\geq 1-2\eta$,
    where \cref{eq-221646} is by \cite[Lemma 12]{AKN98} (see also \cite{kretschmann2008information}), \cref{eq-221647} uses the fact that $\Phi_2$ is a diagonal unitary, $\widehat{V_2}$ is an isometry and $F$ is a unitary, and \cref{eq-221648} uses \Cref{eq:two-weak-tomo,eq:Phi-estimate}.
    Finally, our algorithm outputs the isometry channel corresponding to $\widehat{V'}=\widehat{V_2} \Phi^\dagger F^\dagger$ as the estimate.
\end{proof}

Using the above lemma,
one is ready to prove \Cref{thm-1240018} for isometry channel tomography.

\begin{proof}[Proof of \Cref{thm-1240018}]
    We extend the proof of \cite[Theorem 2.1]{haah2023query} to isometry channel tomography.
    By \Cref{lmm:weak-tomo} (with proper rescaling of $\varepsilon$), it is sufficient to construct a weak isometry channel tomography algorithm that satisfies \Cref{eq:weak-tomo-precond} with $\eta= \frac{1}{6}$.
    The algorithm works as follows:
    \begin{enumerate}
        \item 
            Given queries to $\calV$, we first apply the pure state tomography algorithm in~\Cref{lmm:pure-state-tomo} (taking $d=d_2$ and $\varepsilon_{\max}=\Theta(\varepsilon^2)$ to be determined later) on computational basis input states $\ket{0},\ket{1},\ldots, \ket{d_1-1}$ to obtain estimates $\ket{\widetilde{v_j}}$ of
            $\ket{v_j}=V\ket{j}$ for all $j$ in parallel. We have
            \begin{equation}
                \label{eq:cond-each-vj}
                \ket{\widetilde{v_j}}= \phi_j \sqrt{1-\varepsilon_j} \ket{v_j} + \sqrt{\varepsilon_j} \ket{w_j},
            \end{equation}
            where for each $j=0$ to $d_1-1$, the random variables $\phi_j, \varepsilon_j, \ket{w_j}$ are as in \Cref{lmm:pure-state-tomo}.
        \item 
            Define $\widetilde{V}=\sum_j \ket{\widetilde{v_j}}\!\bra{j}$.
            Suppose $\widetilde{V}$ has the singular value decomposition $\widetilde{V}=U_2 \Lambda U_1$ with $U_1\in \mathbb{U}_{d_1}$ and $U_2\in \mathbb{U}_{d_2}$, respectively.
            Output the quantum channel $\widehat{\calV}$ corresponding to the isometry $\widehat{V}=U_2 \sum_{j=0}^{d_1-1} \ket{j}\!\bra{j} U_1$.
    \end{enumerate}
    It is easy to calculate that the number of queries to $\calV$ in the above algorithm is $O\parens*{d_1d_2 /\varepsilon^2}$. To see that $\widehat{V}$ satisfies \Cref{eq:weak-tomo-precond} in \Cref{lmm:weak-tomo},
    let us prove that
    \begin{equation}
        \label{eq:target-iso-tomo}
        \norm*{V \Phi- \widetilde{V}}_{\textup{op}}\leq \varepsilon/2
    \end{equation}
    for some diagonal unitary $\Phi:\Co^{d_1}\rightarrow \Co^{d_1}$,
    with probability $\geq 0.97 \geq \frac{5}{6}$.
    As long as \Cref{eq:target-iso-tomo} holds,
    the estimate $\widehat{V}$
    in the algorithm above will satisfy
    \begin{equation*}
        \norm*{V\Phi -\widehat{V}}_{\textup{op}}
        \leq \norm*{V\Phi- \widetilde{V}}_{\textup{op}}
        + \norm*{\widetilde{V}-\widehat{V}}_{\textup{op}}
        \leq \varepsilon,
    \end{equation*}
    where we used $\|\widetilde{V}-\widehat{V}\|_{\mathrm{op}} \leq \varepsilon/2$ since once \cref{eq:target-iso-tomo} holds, the operator norm between $\widetilde{V}$ and an isometry is at most $\varepsilon/2$ and thus the differences between the singular values of $\widetilde{V}$ and $1$ are at most $\varepsilon/2$.
    
    Let $W=\sum_{j=0}^{d_1-1} \ket{w_j}\!\bra{j}$, $\Phi=\sum_{j=0}^{d_1-1} \phi_j \ket{j}\!\bra{j}$, $B_1=\sum_{j=0}^{d_1-1} \sqrt{\varepsilon_j}\ket{j}\!\bra{j}$, and $B_2=\sum_{j=0}^{d_1-1} \sqrt{1-\varepsilon_j}\ket{j}\!\bra{j}$.
    Then, $\abs*{\sqrt{1-\varepsilon_j}- 1}\leq \sqrt{\varepsilon_j}$ for all $j=0$ to $d_1-1$ implies
    $\norm*{B_2-I_{d_1}}_{\textup{op}}\leq \norm*{B_1}_{\textup{op}}$,
    where $I_{d_1}=\sum_{j=0}^{d_1-1} \ket{j}\!\bra{j}$.
    By \Cref{lmm:pure-state-tomo} (taking $d=d_2$, where $d_2$ are sufficiently large), we have 
    \begin{equation*}
        \norm*{B_1}_{\textup{op}} \leq \sqrt{\varepsilon_{\max}},
    \end{equation*}
    with probability $\geq 0.99$.
    By the triangle inequality,
    \begin{equation*}
        \norm*{V\Phi- \widetilde{V}}_{\textup{op}} = \norm*{V\Phi (B_2- I_{d_1}) + W B_1}_{\textup{op}}
        \leq \norm*{V}_{\textup{op}}\cdot\norm*{\Phi}_{\textup{op}}\cdot\norm*{B_2- I_{d_1}}_{\textup{op}} + \norm*{W}_{\textup{op}}\cdot\norm*{B_1}_{\textup{op}}
    \end{equation*}
    and therefore
    \begin{equation}
        \label{eq:bound-V-tildeV}
        \norm*{V\Phi- \widetilde{V}}_{\textup{op}}\leq \sqrt{\varepsilon_{\max}} (1+\norm*{W}_{\textup{op}}),
    \end{equation}
    with probability $\geq 0.99$.

    Now we prove that
    \begin{equation}
        \label{eq:bound-W-cw}
        \norm*{W}_{\textup{op}}\leq c_W
    \end{equation}
    for some constant $c_W>0$,
    with probability $\geq 0.98$.
    This will lead to \Cref{eq:target-iso-tomo} by combining with \Cref{eq:bound-V-tildeV} and choosing $\varepsilon_{\max}=\Theta(\varepsilon^2)$ to be sufficiently small.
    
    For each $j=0$ to $d_1-1$, define a quantum state
    \begin{equation}
        \label{eq:def-yj}
        \ket{y_j} = \sqrt{\delta_j} \psi_j \ket{v_j} + \sqrt{1-\delta_j} \ket{w_j},
    \end{equation}
    where $\sqrt{\delta_j}= \abs*{\braket{0}{x_j}}$ is the overlap between a Haar random state $\ket{x_j}\sim \Co^{d_2}$ and the state $\ket{0}$, and $\psi_j\sim [0,2\pi)$ is an uniformly random phase.
    We also require that $\ket{x_0},\ldots,\ket{x_{d_1-1}}$ are independent
    and $\psi_0,\ldots, \psi_{d_1-1}$ are independent.
    Then, $\ket{y_j}\sim \Co^{d_2}$ and $\ket{y_0},\ldots,\ket{y_{d_1-1}}$ are independent. 
    Let $Y=\sum_{j=0}^{d_1-1} \ket{y_j}\!\bra{j}$.
    Using \cite[Theorem 3.4.6, complex version]{Ver18}, $\sqrt{d_2} Y$ has its column vectors being independent sub-gaussian isotropic random in $\Co^{d_2}$,
    and we can bound the maximal singular value of $Y$ with high probability by \cite[Theorem 4.6.1, complex version]{Ver18}:
    \begin{equation}
        \label{eq:bound-Y}
        \norm*{Y}_{\textup{op}}\leq c_Y
    \end{equation}
    for some constant $c_Y>0$, with probability $\geq 0.99$.
    Let $E_1=\sum_{j=0}^{d_1-1} \sqrt{\delta_j}\psi_j \ket{j}\!\bra{j}$ and
    $E_2=\sum_{j=0}^{d_1-1} \sqrt{1-\delta_j} \ket{j}\!\bra{j}$.
    Since 
    \begin{equation*}
        \norm*{W}_{\textup{op}}= \norm*{\parens*{Y - V E_1} E_2^{-1}}_{\textup{op}}\leq
        \parens*{\norm*{Y}_{\textup{op}}+\norm*{V}_{\textup{op}}\cdot \norm*{E_1}_{\textup{op}}}\cdot \norm*{E_2^{-1}}_{\textup{op}},
    \end{equation*}
    by combining with \Cref{eq:def-yj,eq:bound-Y}, we have
    \begin{equation}
        \label{eq:bound-W}
        \norm*{W}_{\textup{op}}
        \leq \parens*{c_Y + 1} \cdot \parens{1-\max_{j} \delta_j}^{-1/2}
    \end{equation}
    with probability $\geq 0.99$.
    As $\sqrt{d_2}\ket{x_j}$ are sub-gaussian (like the case of $\sqrt{d_2}\ket{y_j}$, according to \cite[Theorem 3.4.6, complex version]{Ver18}), $\sqrt{d_2\delta_j}$ are also sub-gaussian by definition, yielding
    \begin{equation*}
        \Pr\bracks*{\sqrt{\delta_j} \leq 0.1} \geq 1- e^{-\Theta(d_2)}.
    \end{equation*}
    By a union bound, for sufficiently large $d_2$ and note that $d_1\leq d_2$,
    $\Pr\bracks*{(1-\max_j \delta_j)^{-1/2}\leq 2}\geq 0.99$,
    and therefore we establish \Cref{eq:bound-W-cw}.
\end{proof}

\end{document}